%% file: main.tex
\documentclass[11pt]{article}

\usepackage[colorlinks=true,linkcolor=blue,citecolor=blue]{hyperref}
\usepackage{amsmath, amssymb, amsthm}
\usepackage{mathtools}
\usepackage[noabbrev,capitalize,nameinlink]{cleveref}
\crefname{equation}{}{}
\usepackage{fullpage}
\usepackage{graphics}
\usepackage{pifont}
\usepackage{tikz}
\usepackage{bbm}
\usepackage[T1]{fontenc}
\DeclareMathOperator*{\argmax}{arg\,max}
\DeclareMathOperator*{\argmin}{arg\,min}
\usetikzlibrary{arrows.meta}

\usepackage{environ}
\usepackage{framed}
\usepackage{url}
\usepackage[linesnumbered,ruled,vlined]{algorithm2e}
\usepackage[noend]{algpseudocode}
\usepackage[labelfont=bf]{caption}
\usepackage{cite}
\usepackage{framed}
\usepackage[framemethod=tikz]{mdframed}
\usepackage{appendix}
\usepackage{graphicx}
\usepackage[textsize=tiny]{todonotes}
\usepackage{tcolorbox}
\allowdisplaybreaks[1]

\newcommand\remove[1]{}

\newtheorem{theorem}{Theorem}
\newtheorem{lemma}{Lemma}[section]
\newtheorem*{lemma*}{Lemma}
\newtheorem{corollary}[lemma]{Corollary}
\newtheorem*{corollary*}{Corollary}

\theoremstyle{definition}

\newtheorem*{theorem*}{Theorem}
\newtheorem{definition}[lemma]{Definition}

\newtheorem*{rem*}{Remark}

\newcommand\R{\mathbb{R}}

\newcommand\Z{\mathbb{Z}}

\newcommand\E{\mathbb{E}}

\newcommand{\eps}{\varepsilon}
\renewcommand{\O}{\widetilde{O}}

\newcommand{\pe}{\preceq}
\newcommand{\se}{\succeq}

\newcommand{\bs}{\backslash}

\newcommand{\hf}{\hat{f}}
\newcommand{\assign}{\leftarrow}
\newcommand{\otilde}{\O}
\renewcommand{\forall}{\mathrm{\text{ for all }}}
\renewcommand{\d}{\delta}

\newcommand{\wpe}{w^+_e}
\newcommand{\wme}{w^-_e}
\newcommand{\upe}{u^+_e}
\newcommand{\ume}{u^-_e}
\newcommand{\rpe}{\rho^+_e}
\newcommand{\rme}{\rho^-_e}
\newcommand{\g}{\nabla}
\newcommand{\diag}{\mathrm{diag}}
\renewcommand{\E}{\mathcal{E}}

\newcommand{\new}{\mathrm{new}}

\newcommand{\up}{u^+}
\newcommand{\um}{u^-}
\renewcommand{\wp}{w^+}
\newcommand{\wm}{w^-}
\newcommand{\rp}{\rho^+}
\renewcommand{\rm}{\rho^-}
\newcommand{\normrho}{\|\rho\|}
\newcommand{\val}{\mathrm{val}}
\newcommand{\Maxflow}{\textsc{Maxflow}}

\newcommand{\RecursivePreconditioning}{\textsc{RecursivePreconditioning}}
\newcommand{\X}{\mathcal{X}}
\newcommand{\T}{\mathcal{T}}

\newcommand{\Center}{\textsc{Center}}
\newcommand{\Progress}{\textsc{Progress}}
\newcommand{\ControlCongestion}{\textsc{ControlCongestion}}
\newcommand{\ComputeWeights}{\textsc{ComputeWeights}}
\newcommand{\ReduceWeights}{\textsc{ReduceWeights}}
\newcommand{\PerfectCenter}{\textsc{PerfectCenter}}

\newif\ifrandom
\randomtrue

\newcommand{\maxflowruntime}{m^{11/8+o(1)}U^{1/4}}

\newcommand{\defeq}{\stackrel{\mathrm{\scriptscriptstyle def}}{=}}

\newcommand{\poly}{{\mathrm{poly}}}
\newcommand{\err}{\frac{1}{2^{\poly(\log m)}}}

\newcommand{\todolater}[1]{}

\newcommand{\etal}{\textit{et~al.}}

\crefname{algocf}{Algorithm}{Algorithms}

\author{
Yang P. Liu \\
Stanford University \\
\texttt{yangpliu@stanford.edu}
\thanks{Research supported by the U.S.
Department of Defense via an NDSEG fellowship.} 
\and
Aaron Sidford \\
Stanford University \\
\texttt{sidford@stanford.edu}
\thanks{Research supported by NSF CAREER Award CCF-1844855.}
}

\begin{document}

\title{Faster Energy Maximization for Faster Maximum Flow}

\begin{titlepage}
\clearpage\maketitle
\thispagestyle{empty}

\input{abstract.tex}

\end{titlepage}

\newpage

\input{intro2.tex}
\input{prelim.tex}
\input{approach.tex}
\input{setup.tex}
\input{weight.tex}

\section{Acknowledgements}
\label{sec:acknowledgements}

We thank Arun Jambulapati, Michael B. Cohen, Yin Tat Lee, Jonathan Kelner, Aleksander M\k{a}dry, and Richard Peng for helpful discussions.

{\small
\bibliographystyle{abbrv}
\bibliography{refs}}

\appendix

\input{appprelim.tex}
\input{proofs.tex}
\input{approximate.tex}
\input{numerical.tex}

\end{document}

%% file: abstract.tex
\begin{abstract}
	
In this paper we provide an algorithm which given any $m$-edge $n$-vertex directed graph with integer capacities at most $U$ computes a maximum $s$-$t$ flow for any vertices $s$ and $t$ in $\maxflowruntime$ time with high probability. This running time improves upon the previous best of $\otilde(m^{10/7} U^{1/7})$ (M\k{a}dry 2016), $\otilde(m \sqrt{n} \log U)$ (Lee Sidford 2014), and $O(mn)$ (Orlin 2013) when the graph is not too dense or has large capacities.

We achieve this result by leveraging recent advances in solving undirected flow problems on graphs. We show that in the maximum flow framework of (M\k{a}dry 2016) the problem of optimizing the amount of perturbation of the central path needed to maximize energy and thereby reduce congestion can be efficiently reduced to a smoothed $\ell_2$-$\ell_p$ flow optimization problem, which can be solved approximately via recent work (Kyng, Peng, Sachdeva, Wang 2019). Leveraging this new primitive, we provide a new long-step interior point method for maximum flow with faster convergence and simpler analysis that no longer needs global potential functions involving energy as in previous methods (M\k{a}dry 2013, M\k{a}dry 2016).

\end{abstract}

%% file: intro2.tex
\section{Introduction}
\label{sec:intro}
The maximum flow problem is one of the most well studied problems in combinatorial optimization and encompasses a broad range of cut, matching, and scheduling problems \cite{CLRS01}. Given a directed graph $G = (V, E)$ with $m = |E|$ edges, $n = |V|$ vertices, integer capacities $u \in \Z^E_{\geq 0}$, and vertices $s, t \in V$, the maximum flow problem asks to find a $s$-$t$ flow $f \in \R^E_{\geq 0}$ which sends the maximum number of units of flow from $s$ to $t$ without violating capacity constraints, i.e. satisfies that $f_e \leq u_e$ for all $e \in E$ (see \cref{sec:prelim}).

The main result of this paper is a new algorithm for solving maximum flow with high probability in time $\maxflowruntime = m^{3/2 - 1/8 + o(1)} U^{1/4}$, improving upon the previous best running times of $\O(m^{3/2 - 1/14} U^{1/7})$ \cite{Madry16}, $\O(m \sqrt{n})$ \cite{LS14}, and $O(mn)$ \cite{Orlin13} whenever the graph is not too dense or has large capacities. Our result immediately implies comparable improvements for the problems discussed above, including computing maximum matchings in bipartite graphs.

\subsection{Beyond the Laplacian Paradigm}

Our results are motivated by the particularly challenging problem of improving the running time for solving the maximum flow problem in unit capacity graphs, i.e. when $u_e = 1$ for all $e \in E$. This problem is known to be equivalent to computing a maximum set of disjoint $s$-$t$ paths \cite{Menger27,LRS13}, is sufficient for solving maximum cardinality bipartite matching, and has had few running time improvements in its long history. See \cref{sec:previous_work} for a more comprehensive discussion.

The state-of-the-art running times for unit capacity maximum flow of $\O(m^{10/7})$ \cite{Madry13} and $\O(m\sqrt{n})$ \cite{LS14} are instances of what is known as the \emph{Laplacian Paradigm}. Each result leverages a seminal result of Spielman and Teng in 2004 \cite{ST04} which showed that Laplacian systems could be solved in nearly linear time (see \cref{sec:prelim,sec:lap}). The results of \cite{Madry13,LS14} cast the maximum flow problem as a linear programming problem and applied interior point methods (IPMs), i.e. the state-of-the-art optimization methods for provably solving linear programs \cite{LS14,LeeS15,CohenLS19,LeeS19arXiv}; IPMs reduce solving linear programs to solving a sequence of linear systems, which are Laplacians in the special case of the maximum flow linear program. Tailoring these IPMs to the structure of maximum flow problem \cite{Madry13} and improving upon classic IPM methods \cite{LS14} then 
yield the previous state-of-the-art running times for unit capacity maximum flow.

Consequently, in the previous state-of-the-art maximum flow algorithms, much of the particular structure of maximum flow is abstracted away through the application of Laplacian system solvers in IPM frameworks. \cite{Madry13, Madry16} does provide additional graph-specific analysis and techniques, however these methods are somewhat local and do not perform the same type of global graph processing that goes into solving Laplacian systems.

Though the Laplacian paradigm, i.e. using Laplacian system solvers to solve graph problems efficiently, has been widely successful \cite{SpielmanS08,KM09,CKMST11,OrecchiaSV12,LRS13,KLOS14,MST15,ZhuLO15,LeeS15a,CohenMSV17,Lee017,Schild18}, they appear to be a sub-optimal primitive for solving maximum flow in certain cases. In another breakthrough of Christiano, Kelner, M\k{a}dry, Spielman, and Teng in 2011 \cite{CKMST11}, it was noted that solving Laplacian systems correspond to a particular flow optimization problem known as \emph{electric flow}, i.e. minimizing weighted $\ell_2$-norms of flows routing a particular demand; this is in contrast to solving undirected maximum flow, which is a particular $\ell_\infty$ minimization problem.  \cite{CKMST11} leveraged this insight and combined it with new optimization techniques to improve the running time for approximately solving maximum flow in undirected graphs, laying the groundwork for \cite{Madry13, Madry16}.

Interestingly, it was more recently shown that directly building crude generalizations of Laplacian system solvers for $\ell_\infty$ flow minimization, namely congestion approximators \cite{Sherman13} and oblivious routings \cite{KLOS14}, give almost linear time algorithms for approximating undirected maximum flow. This spawned a line of research on building new graph primitives and has since lead to even faster undirected maximum flow algorithms, nearly linear time solvers for directed graph variants of Laplacians, and more (see \cref{sec:previous_work}). These results raise the tantalizing possibility of using these new graph primitives to improve the running time for solving maximum flow and move beyond the Laplacian paradigm, i.e. the direct use of Laplacian system solvers. However despite extensive research obtaining such improvements for solving unit-capacity maximum flow has been elusive.

In this paper we overcome this barrier and show how to leverage a recent almost linear time flow primitive, namely the undirected graph smoothed $\ell_2$-$\ell_p$ flow algorithm of \cite{KPSW19}, in conjunction with Laplacian system solvers to improve upon the running time of \cite{Madry13, Madry16}. By leveraging this stronger primitive we not only obtain a faster running time, but in our opinion a slightly simpler algorithm than was obtained by previous interior point frameworks. We hope this work opens the door to further merging these diverged lines of research on directed and undirected flow algorithms and the varied optimization techniques which underlie them.

\subsection{Our Results}

The main result of this paper is a faster running time for solving maximum flow.
\begin{theorem}[Maximum Flow] \label{thm:main}
There is an algorithm which given a maximum flow problem instance on an integer capacitated graphs with $n$ vertices, $m$ edges, and maximum capacity $U$ computes a maximum flow with high probability in time $\maxflowruntime$. 
\end{theorem}
This improves on the $\O(m^{10/7}U^{1/7})$ time algorithm of \cite{Madry16} as long as $U \le m^{1/2-\eps}$ for some $\eps > 0.$ In the case that
$U \ge m^{1/2}$ both the algorithm of \cite{Madry16} and \cref{thm:main} have runtime $\O(m^{3/2})$, which is already known through the work of Goldberg and Rao \cite{GR98}. Hence, we assume $U \le \sqrt{m}$ throughout the paper. An immediate corollary of \cref{thm:main} is the following result on solving maximum cardinality bipartite matching.
\begin{corollary}[Bipartite Matching] \label{cor:matching}
	There is an algorithm which given a bipartite graph with $n$ vertices and $m$ edges computes a maximum matching with high probability in time $m^{11/8 + o(1)}$. 
\end{corollary}
We achieve this result by providing a new interior point method, broadly inspired by \cite{Madry16}, which allows us to leverage recent advances in computing smoothed $\ell_2$-$\ell_p$ norm minimizing flows to achieve faster running times. Further, we primarily use the smoothed $\ell_2$-$\ell_p$ flow minimization procedure to solve a particular type of energy maximization or congestion minimization problem that also may be of independent interest and useful for further results (see \cref{lemma:rhoinf} and \cref{sec:proofrhoinf}).

\subsection{Previous Work}
\label{sec:previous_work}

Given the well-studied nature of the maximum flow problem, it is impossible to comprehensively cover its storied history in this short section. Instead, we cover the lines of research on maximum flow most relevant to the results of this paper. Throughout, we let $m$ denote the number of edges, $n$ the number of vertices, and $U$ the maximum integer capacity in the graph. Further, we use $\tilde{O}$ to hide polylogarithmic factors in $m, n$, and $U$ and for simplicity we ignore these factors in the comparison of algorithms and do not distinguish between deterministic and randomized algorithms.

\textbf{Unit Capacity Graphs:} Obtaining running time improvements for exact maximum flow is a notoriously difficult open problem and obtaining improvements in the the special case of unit capacity graphs is particularly challenging. The problem was first shown to be polynomial time solvable in the work of Ford and Fulkerson \cite{FF56} and a running time of $\O( \min \{ m^{3/2}, m n^{2/3}  \} )$ was achieved in the classic work of Karzanov and Even-Tarjan \cite{Karzanov73,ET75} in the 1970s. Though there has been extensive progress on improving the running time of exact maximum flow on unit capacity graphs in varied special cases (e.g. planar graphs \cite{Reif83,Federickson87,Wulff10,IS10,KS19} and undirected graphs \cite{Karger97,Karger98a,Karger98b,Karger99}) and more general settings (e.g. multicommodity flow \cite{Fleischer00,KapoorV96,GargK98,Karakostas08,Madry10,KelnerMP12,KLOS14,LeeS15,Sherman17} and minimum cost flow \cite{GT89,DS08,CohenMSV17}), there have only been two improvements to this bound since. The work of \cite{Madry13} provided a $\O(m^{10/7})$ time algorithm and \cite{LS14} provided a $\O(m \sqrt{n})$ time algorithm for this problem. In the special case of bipartite matching, the only further improvement was the work of Hopcroft and Karp \cite{HK73} which showed the the problem could be solved in $O(m\sqrt{n})$, well in advance of \cite{LS14}.

\textbf{Capacitated Graphs:} In the 1990s there was a long line of improved algorithms for solving maximum flow on capacitated graphs (see \cite{GR98} for a more comprehensive discussion). In the case of weakly-polynomial time algorithms (i.e. those that depend poly-logarithmically on $U$) these results culminated in the seminal result of Goldberg and Rao in 1998 \cite{GR98} which showed that the problem could be solved in $\O( \min \{ m^{3/2}, m n^{2/3}  \} \log U)$. This was then improved to $\otilde(m \sqrt{n} \log U)$ by \cite{LS14}. In the case of strongly-polynomial time algorithms (i.e. those that do not depend on $U$ at all) the state-of-the-art is \cite{Orlin13} which showed that the problem can be solved in $O(mn)$. Further, \cite{Madry16} simplified components of \cite{Madry13} and obtained a $\O(m^{10/7} U^{1/7})$ time algorithm for maximum flow.

\textbf{Undirected Graphs:} Though the state-of-the-art running times for exact maximum flow on directed graphs has been difficult to improve, the last decade has witnessed numerous advances in solving maximum flow on undirected graphs. One line of work \cite{Karger97,Karger98a,Karger98b,Karger99} achieved improved randomized combinatorial algorithms for exact maximum flow on undirected graphs. In another line of work techniques from combinatorial and continuous optimization were combined to obtain faster algorithms for approximately solving maximum flow on undirected graphs \cite{CKMST11,Sherman13,KLOS14,Peng16,Sherman17,ST18} and Eulerian directed graphs \cite{EMPS16,IT18}. For these problems $\eps$-approximate solutions can now be computed in nearly linear $\O(m/\eps)$ time and even faster in some cases \cite{ST18}.

\textbf{Undirected Flow Problems:} Many of the above advances in solving maximum flow hinged on advances in solving related flow problems. The algorithms of \cite{CKMST11,Madry13,LS14,Madry16} all use Laplacian system solvers to compute electric flow, i.e. weighted $\ell_2$ norm minimizing flows on graphs. There have been numerous advances in solving Laplacian systems over the past decade \cite{ST04, KMP10, KMP11, KOSZ13, CKMPPRX14, KLPSS16, KS16} and more recently advances in solving $\ell_p$ norm variants of the problem \cite{AKPS19,KPSW19}. Our results hinge on a recent result  \cite{KPSW19} on computing particular combinations of the $\ell_2$ and $\ell_p$ norm of graphs efficiently.

\subsection{Paper Organization}
\label{sec:organization}
In \cref{sec:prelim} we give the preliminaries. In \cref{sec:approach} we give an overview for our algorithms. In \cref{sec:setup} we set up our main IPM framework. In \cref{sec:weight} we give a framework for understanding weight changes in the interior point method and give our main improved weight change algorithms. Finally, in \cref{sec:proofrhoinf} we prove \cref{lemma:rhoinf} which gives an efficient algorithm for controlling congestion via weight change and in \cref{sec:discussion} we conclude with possible simplifications and open problems.

Many proof details are deferred to the appendix, which is organized as follows. In \cref{sec:appprelim} we give additional preliminaries for convex optimization and Laplacian solvers. In \cref{sec:proofs,sec:madryweightproofs} we provide missing proofs. In \cref{sec:approximate,sec:numerical} we discuss issues arising from the approximate nature of the Laplacian and smoothed $\ell_2$-$\ell_p$ flow solvers, and numerical stability of IPMs.

%% file: prelim.tex
\section{Preliminaries}
\label{sec:prelim}

\paragraph{General notation}
We let $\R_{\ge\alpha}^m$ denote the set of $m$-dimensional real vectors which are entrywise at least $\alpha$. For a vector $v \in \R^m$ and real number $p \ge 1$ we define $\|v\|_p$, the $\ell_p$-norm of $v$, as $\|v\|_p = \left(\sum_{i=1}^m |v_i|^p \right)^{1/p}.$ We define $\|v\|_\infty = \lim_{p \to \infty} \|v\|_p = \max_{i=1}^m |v_i|.$ Throughout the vectors $0$ and $1$ denote the vectors where all coordinates are $0$ and $1$ respectively. For a vector $v \in \R^m$ let $V = \diag(v)$ denote the $m \times m$ diagonal matrix with the entries of $v$ on the diagonal.

\paragraph{Graphs} Throughout this paper, in the graph problems we consider, we suppose that there are both upper and lower capacity bounds on all edges. We let $G$ be a graph with vertex set $V$, edge set $E$, and upper and lower capacities $\up_e \ge 0$ and $\um_e \ge 0$ respectively on edge $e$. We use $U$ to denote the maximum capacity of any edge, so that $\max\{\up_e, \um_e\} \le U$ for all edges $e$. We let $n$ denote the number of vertices $|V|$, and let $m$ denote the number of edges $|E|$. Further we view undirected graphs as directed graphs with $\upe = \ume$ by arbitrarily orienting the edges.

\paragraph{The Maximum Flow Problem} 
Given a graph $G = (V,E)$ we call any assignment of real values to the edges of $E$, i.e. $f \in \R^E$, a \emph{flow}. For a flow $f \in \R^E$, we view $f_e$ as the amount of flow on edge $e$ and if $f_e > 0$ we interpret this as sending $f_e$ units in the direction of the edge orientation, and if $f_e < 0$ we interpret this as sending $|f_e|$ units in the direction opposite the edge orientation. 

We say that  $f \in \R^E$ is a \emph{$d$-flow} if it routes demands $d \in \R^V$, meaning that the net flow into each vertex $v \in V$ is given by $d_v$, i.e.  $\sum_{e\in E^+(v)}f_e - \sum_{e\in E^-(v)}f_e = d_v$ where $E^+(v)$ and $E^-(v)$ denote the incoming and outgoing edges to $v$ respectively.
The \emph{incidence matrix} for a graph $G$ is an $m \times n$ matrix $B$, where the row corresponding to edge $e = (u,v)$ has a $1$ (respectively $-1$) in the column corresponding to $v$ (respectively $u$). Note that $f \in \R^E$ is a $d$-flow if and only if $B^Tf=d.$

A specific flow we consider is an $ab$-flow, where $a$ is the single source, and $b$ is the single sink. An $ab$-flow is a flow which routes demand vector $d = t\chi_{ab}$ for some real number $t \ge 0$ where $\chi_{ab} = 1_b - 1_a$ is a demand vector with a $1$ in position $a$ and $-1$ in position $b$. When $a$ and $b$ are implicit, we write $\chi = \chi_{ab}$. We say that a $d$-flow $f$ is \emph{feasible} in $G$ if
$-\ume \le f_e \le \upe \forall e \in E,$ so that $f$ satisfies the capacity constraints. We define the \emph{maximum flow problem} as the problem of given a graph $G$ with upper and lower capacities $\up$ and $\um$, and source and sink vertices $a$ and $b$, to compute a maximum feasible $ab$-flow. We denote the maximum value as $t^*$.

\paragraph{Electric Flow and Laplacian Systems} 
Our algorithms make heavy use of electric flows and Laplacian system solvers. Here, we review their definitions and important properties.

Let $G$ be a graph, and let $r \in \R_{>0}^E$ be a vector of edge \emph{resistances}, where the resistance of edge $e$ is denoted by $r_e.$ For a flow $f \in \R^E$ on $G$, we define the \emph{energy} of $f$ to be $\E_r(f) \defeq f^TRf = \sum_{e\in E} r_ef_e^2,$ where $R = \diag(r)$. For a demand $d$, we define the \emph{electrical $d$-flow} $\hf$ to be the $d$-flow which minimizes energy, i.e. $\hf = \argmin_{B^Tf=d} \E_r(f).$
As energy is strictly convex, this flow is unique.

The \emph{Laplacian} of the graph $G$ with resistances $r$ is defined as $L \defeq B^TR^{-1}B$. Then the electric $d$-flow $\hf$ is given by the formula $\hf = R^{-1}BL^\dagger d$, where $L^\dagger$ denotes the \emph{Moore-Penrose pseudoinverse} of $L$. We define \emph{potentials} $\phi \defeq L^\dagger d$, so that $\hf = R^{-1}B\phi$. There is a long line of work \cite{ST04, KMP10, KMP11, KOSZ13, CKMPPRX14, KLPSS16, KS16} building up from early results of Vaidya \cite{Vaidya99} and Spielman and Teng \cite{ST03} which give algorithms to solve such Laplacian systems to compute $\phi$ and $\hf$ in nearly linear time. We defer the precise theorem statements and other results to \cref{sec:lap}.

%% file: approach.tex
\section{Overview of Approach}
\label{sec:approach}

Our maximum flow algorithm builds on the framework of M\k{a}dry \cite{Madry16}. We consider the setup described in \cref{sec:prelim} and design algorithms that maintain a flow  $f \in \R^E_{> 0}$, a parameter $t > 0$, and weights $\wp, \wm \in \R_{\ge1}^m$ such that $B^T f = t \chi$ and $f$ is an approximate solution to the following optimization problem.
\begin{equation}
\label{eq:central_path_formula}
\min_{B^Tf=t\chi} V(f)
\enspace
\text{ where }
\enspace
V(f) \defeq -\sum_{e\in E} \left(\wp_e \log(\up_e-f_e) + \wm_e \log(\um_e+f_e)\right)
~.
\end{equation}
Here, $V(f)$ is known as the \emph{weighted logarithmic barrier function} and penalizes how close $f$ is to violating the capacity constraints and $t$ controls how much flow is sent from $a$ to $b$.

Broadly speaking, IPMs iterate towards optimal solutions by trading off the quality of the current iterate, e.g. $B^T f = t\chi$, and the proximity of the point to the constraints, e.g. $V(f)$. The previous maximum flow IPMs \cite{Madry13,LS14,Madry16} and the methods of this paper all follow this template. In particular, \cite{Madry16} and our method alternate between taking Newton steps to improve the optimality of $f$ in solving \cref{eq:central_path_formula} (known as \emph{centering}) and computing a new flow and new weights which approximately solve \cref{eq:central_path_formula} for larger values of $t$ (which we call a \emph{progress step}).

Centering and progress steps in both \cite{Madry16} and our method involve computing an electric flow for edge resistances induced by the Hessian of $V(f)$. Centering is achieved by using this flow to add a circulation to the current flow and progress steps are made by adding a multiple of the $\chi$-electric flow to the current flow. \cite{Madry16} showed that by a minor modification to the graph, called preconditioning (see \cref{sec:precondition}), it can be easily shown that even without weight changes alternating between such centering and progress steps suffices to obtain a $\otilde(m^{3/2} \log U)$ time algorithm. Obtaining such a $\otilde(m^{3/2} \log U)$ time algorithm for solving maximum flow by IPMs and Laplacian system solvers is easily achieved by classic interior point theory \cite{Renegar88,DS08} and state-of-the-art maximum flow IPMs \cite{LS14,Madry16} all improve by more carefully changing weights and designing progress steps.

In particular, \cite{Madry16} achieves its $\O(m^{10/7}U^{1/7})$ runtime by considering a more fine-grained analysis of how much progress a progress step can make. The paper considers the congestion vector, defined as the ratio of amount of $\chi$-electric flow induced by Hessian resistances to the residual capacity of each edge (see \cref{sec:congest}). In other words, the congestion vector bounds how quickly adding the electric flow to the current flow would saturate each edge. The work of \cite{Madry13,Madry16} both note that the $\ell_\infty$-norm of this vector bounds how much flow can be added without violating a constraint and the $\ell_4$-norm bounds how much flow can be added such that centering can still be performed in $\otilde(m)$ time. This is reminiscent of the classic predictor-corrector IPM analysis \cite{Mehrotra92,YTM94}.

Further, \cite{Madry13,Madry16} note that when the $\ell_\infty$ and $\ell_4$ norm of the congestion vector are large then increasing the resistances of the congested edges would dramatically increase the energy of the electric flow. This was a key insight of \cite{CKMST11}, albeit in a different setting. Consequently, \cite{Madry16} repeatedly increases the weights of congested edges (called \emph{boosting}) until the congestion vector is sufficiently small. By using electric energy as a global potential function and analyzing how it evolves over progress, centering, and boosting steps, the paper is able to control the total amount of weight change and total runtime of boosting steps necessary to reduce the $\ell_\infty$ and $\ell_4$ norms of the congestion vector. Carefully trading off these quantities yields the running time of \cite{Madry16}.

To improve on \cite{Madry16} we bypass the cost of computing such repeated weight increases and ultimately alleviate the need for such a global energy potential amortized over the life of the algorithm. We treat the problem of computing a small amount of resistance increase needed to increase energy (and thereby decrease congestion) as a self-contained optimization problem. We analyze the budget-constrained energy maximization problem and use this to show that a small amount of weight change suffices to obtain good bounds on the $\ell_\infty$ norm of congestion.

Further, we show that this budget-constrained maximization problem can be solved approximately in almost-linear time, yielding congestion-reducing weight changes, by leveraging recent smoothed $\ell_2$-$\ell_p$ flow solvers of \cite{KPSW19}. The connection between smoothed $\ell_2$-$\ell_p$ flows and energy maximization in the context of implementing a low-width multiplicative weights oracle as in \cite{CKMST11} had been noted by \cite{KPSW19} and we formally show that this machinery applies in our setting. This involves a variety of convex-optimization reduction techniques and analysis provided in \cref{sec:proofrhoinf}.

We leverage this congestion reduction procedure to design a new progress step. We compute congestion reducing weights and then look at the electric flow induced and show that after taking a progress step it is possible to undo some of the weight change we performed. This yields better bounds on the total weight change performed, but unfortunately this step still has too large a congestion for centering to be efficient. To overcome this, we show that further weight changes can be performed after the progress step so that centering can be performed efficiently.

Putting this all together, we show that for all $\eta$, provided that there is at least an additive $m^{1/2-\eta}$ units of flow still left to route, we can route a $m^{-1/2+\eta}$-multiplicative fraction of the remaining flow in almost linear time per step, while increasing the weights by at most $m^{5\eta}U$, where $U$ is the maximum capacity of any edge. Picking $\eta = \log_m(m^{1/8-o(1)}U^{-1/4})$ to trade-off these two cases and computing the remaining flow by augmenting paths \cite{LRS13} ultimately yields our final result.

We believe that our algorithm makes two broader contributions to the theory of IPMs and maxflow beyond improving its runtime. First, our algorithm uses both $\ell_2$ norm and smoothed $\ell_2$-$\ell_p$ norm flow minimization procedures (Laplacian system solvers and \cite{KPSW19}). We hope that this use of undirected graph tools beyond Laplacian system solving to solve directed maximum flow will open the door to further advances by combining these two lines of research. Second, our algorithm is the first that breaks the $m^{3/2}$ barrier for maximum flow without using energy as a potential function over the life of an IPM. Interestingly, we show that even simpler variants of our method could achieve this goal (see \cref{sec:discussion}) and we hope this insight will also aid further research.

%% file: setup.tex
\section{Interior Point Method Setup}
\label{sec:setup}
It is known how to reduce directed maxflow to undirected maxflow with linear overhead, and we give a proof in \cref{proofs:undirreduc}. Throughout this section, we assume our graph is undirected, so that $\upe = \ume$. All other proofs in this section are deferred to \cref{sec:proofs}.
\subsection{The Central Path}
\label{sec:centralpath}
Here we give our interior point framework, which is broadly inspired by \cite{Madry16}, and introduce important notation.\footnote{Beyond the major algorithmic differences discussed in \cref{sec:approach} we depart from the IPM framework in \cite{Madry16} in several subtle ways to simplify our presentation and analysis (rather than for algorithmic necessity). We use weights instead of duplicating edges. We measure coupling in the inverse norm of the Hessian instead of an approximation. We use a two-sided congestion vector to make the weighted $2$-norm of the congestion vector exactly equal to the energy of the induced electric $\chi$-flow; in \cite{Madry16} these quantities were only a $2$-approximation of each other.}
As discussed in \cref{sec:approach}, we  work with the setup of the maximum flow problem defined in \cref{sec:prelim} and for a parameter $t$ and weight vectors $\wp, \wm \in \R_{\ge1}^m$, consider the optimization problem \eqref{eq:central_path_formula}. Assuming that there is a feasible $t\chi$-flow, optimality conditions give that the gradient of $V$ at the optimum $f^*_{t,w}$ of \eqref{eq:central_path_formula} is perpendicular to the kernel of $B^T$, i.e. there is a  dual vector $y \in \R^V$ such that
$B y = \g V(f^*_{t,w})$. Consequently, for parameter $t$ and weight vectors $\wp, \wm$ we say that a dual vector $y \in \R^V$ and flow $f$ are on the \emph{weighted central path} if and only if
\begin{align}
B^Tf
= t\chi 
\enspace
\text{ and }
\enspace
[By]_e
=  [\g V(f)]_e = \frac{\wp_e}{\up_e-f_e} - \frac{\wm_e}{\um_e+f_e}
\text{ for all } e\in E
 \label{eq:centralpath}
\end{align}
For simplicity, we write $w = (w^+, w^-) \in \R^{2E}_{\ge1}$, where we define $\R^{2E} \defeq \{(x,y) : x,y \in \R^E \}.$

For a flow $f$ and weights $w$ define the \emph{resistances} induced by $f$ and $w$ as
\begin{equation}r_e \defeq \frac{\wp_e}{(\up_e-f_e)^2} + \frac{\wm_e}{(\um_e+f_e)^2} \label{eq:defr}\end{equation}
and $R = \diag(r).$ $R$ is the Hessian of $V$ at $f$, i.e. $R = \g^2 V(f).$ Here $R$ implicitly depends on $f$.
We define the \emph{slacks} $s_e$ as $s_e = \min\{\upe-f_e,\ume+f_e\}.$ Note that $1/s_e^2 \le r_e.$

Our algorithm maintains points $(f,y,w)$ that are approximately on the central path. We quantify this through the \emph{gaps} $g \in \R^E$ and \emph{coupling} as follows.
\begin{equation} g_e \defeq (y_v-y_u) - \left(\frac{\wpe}{\upe-f_e} - \frac{\wme}{\ume+f_e} \right)
\text{ for all edges } e = (u,v) \in E ~. \label{eq:gg}
 \end{equation}
\begin{definition}[Coupling]
\label{def:coupling}
For a point $(f,y,w)$ and parameter $t$, let $g \in \R^E$ be defined as in \cref{eq:gg}. We say that edge $e$ is \emph{$\zeta_e$-coupled} if $|g_e| \le \zeta_e.$ We say that the point \emph{$(f,y,w)$ is $\gamma$-coupled for parameter $t$} if $B^Tf=t\chi$ and $\|g\|_{R^{-1}} \le \gamma$.
\end{definition}
When $t$ is clear from context, we simply say that a point $(f,y,w)$ is $\gamma$-coupled. We let $F_t \defeq t^*-t$ denote the difference between the value of the maxflow and the value of the current flow.

Our algorithm initializes a $0$-coupled solution for $t = 0, \wp = \wm = 1$. This is trivial to compute as for undirected graphs, i.e. where $\up = \um$, it is the case that $f = y = 0$ is $0$-coupled for $t = 0$. After initialization, our algorithm then follows the following general template.
\begin{enumerate}
\item \textbf{Progress step.} We start with a point $(f,y,w)$ that is $0$-coupled. The parameter $t$ is increased (say to $t^\new$), the weight vector is changed to $w^\new$, and we compute a new point $(f^\new,y^\new,w^\new)$ that is $\frac{1}{100}$-coupled with parameter $t^\new$. In increasing $t$ to $t^\new$, we will have multiplicatively decreased the value of $F_t.$
\item \textbf{Centering step.} We take the $\frac{1}{100}$-coupled point above and compute a new $0$-coupled point with the same parameter $t$ and weights $w$.
\end{enumerate}

\subsection{Congestion Vector and Norms}
\label{sec:congest}
In this section, we define the congestion vector and its norm, which appear in the analysis of the progress step. Consider a graph $H$ with the same edge set as $G$, but with edge resistances $r_e$ defined in \cref{eq:defr}. Let $\hf$ be the electric $\chi$-flow in $H$. We call $\hf$ the electric $\chi$-flow induced by the point $(f,y,w)$. Define the forward and backwards \emph{congestion vectors} $\rp, \rm \in \R^E$ as 
\[
\rp_e = \frac{|\hf_e|}{\up_e-f_e}
\text{ and }
\rm_e = \frac{|\hf_e|}{\um_e+f_e} \text{ for all } e\in E.
\] 
Define $\rho \in \R^{2E}$ as $\rho = (\rp, \rm)$ and  for $k \in [1,\infty]$, define the weighted $\ell_k$-norm of $\rho$ as
\[ \|\rho\|_{w,k} \defeq \left(\sum_{e \in E} \wp_e |\rp_e|^k + \wm_e |\rm_e|^k\right)^{1/k} \text{   and   } \|\rho\|_{w,\infty} \defeq \lim_{k\to\infty} \|\rho\|_{w,k} = \max_e\max\{\rpe,\rme\}. \]
We can relate $\|\rho\|_{w,2}^2$ to the electric energy of the flow.
\begin{lemma}
\label{lemma:rhoenergy}
We have that $\wp_e |\rp_e|^2 + \wm_e |\rm_e|^2 = r_e\hf_e^2$ for all edges $e$ and $\|\rho\|_{w,2}^2 = \E_r(\hf)$.
\end{lemma}

\subsection{Following the Central Path}
\label{sec:follow}
In this section, we explain the effect of progress and centering steps on how coupled our point $(f,y,w)$ is. We start by analyzing a progress step.
\begin{lemma}
\label{lemma:progress}
Let $(f,y,w)$ be a $0$-coupled point for parameter $t$. For any $\delta \le \frac{1}{10\|\rho\|_{w,\infty}}$ there is an algorithm that in $\O(m)$ time computes a point $(f^\new,y^\new,w)$ such that $f^\new$ is a $(t+\delta)\chi$-flow and
\begin{enumerate}
\item $f^\new = f + \d\hf$ where $\hf$ is the electric $\chi$-flow with resistances $r$ as defined in \cref{eq:defr}.
\item Edge $e$ is $5\d^2\left(\frac{\wpe|\rpe|^2}{\upe-f_e} + \frac{\wme|\rme|^2}{\ume+f_e}\right)$-coupled for all edges $e$.
\item The point $(f^\new,y^\new,w)$ is $10\d^2\normrho_{w,4}^2$ coupled.
\end{enumerate}
\end{lemma}
Note that the point $(f^\new,y^\new,w)$ computed in \cref{lemma:progress} is still $\frac{1}{100}$-coupled for the choice $\d = \frac{1}{100\normrho_{w,4}}$. A major component of our new algorithm for maxflow is a new method for performing additional weight changes after taking a step using \cref{lemma:progress}. We introduce this in \cref{sec:weight}.

There is a $\O(m)$ time algorithm that takes a $\gamma$-coupled point $(f,y,w)$ for parameter $t$ and returns a $10\gamma^2$-coupled point $(f^\new,y^\new,w)$ for parameter $t$. This allows us to recenter our point $(f,y,w)$.
\begin{lemma}[Algorithm \Center]
\label{lemma:centering}
There is a $\O(m)$ time algorithm, $\Center(f,y,t,w)$, that given point $(f,y,w)$ that is $\gamma$-coupled for parameter $t$ and $\gamma \le \frac{1}{100}$ returns a point $(f^\new,y^\new,w)$ that is $10\gamma^2$-coupled for parameter $t$.
\end{lemma}
Applying \cref{lemma:centering} $\O(1)$ times allows us to get to a $\frac{1}{\poly(m)}$-centered point $(f,y,w)$. Consequently, in much of our analysis we assume that from a $1/100$-coupled point we can obtain a $0$ coupled point in $\otilde(m)$ time. See \cref{sec:numerical} for a discussion of how to correct the analysis without this assumption.

\subsection{Preconditioning}
\label{sec:precondition}
Combining \cref{lemma:progress} and \cref{lemma:centering}, we see that we can increment our total flow by $\Omega(\normrho_{w,4}^{-1})$ in $\O(m)$ time. Therefore, it is fruitful to understand $\|\rho\|_{w,k}$ for various $k$. As was done in \cite{Madry16}, we modify or \emph{precondition} our graph to obtain such bounds on the congestion vector. Here we describe the preconditioning procedure and state the bounds. 

To precondition our undirected graph $G$, we add $m$ undirected edges of capacity $2U$ between source $a$ and sink $b$. This increases the maximum flow value by $2mU.$ Throughout the remainder of the paper, we say that the graph $G$ is preconditioned if it is undirected and we have added these edges.  Now, we show the following lemma. Here, we have defined $w \in \R^{2E}$ as $w = (\wp, \wm)$.
\begin{lemma}
\label{lemma:rho2control}
\label{cor:rho2morecontrol}
\label{cor:rho2control}
Let $(f,y,w)$ be $\frac{1}{100}$-coupled for parameter $t$ in preconditioned graph $G$. Then we have that $\|\rho\|_{w,2}^2 \le 50\|w\|_1^3 m^{-2} F_t^{-2}$. If $\|w\|_1 \le 3m$ then we have that $\|\rho\|_{w,2}^2 \le 1500m F_t^{-2}.$ Additionally, if $F_t \ge m^{1/2-\eta}$ then we have that $\|\rho\|_{w,2} \le 40m^\eta$.
\end{lemma}
At the start of the algorithm, as we initialized $\wp = \wm = 1$, we have $\|w\|_1 = 2m.$
To apply \cref{lemma:rho2control} we need to maintain the invariant that $\|w\|_1 \le 3m$ throughout.

Combining these pieces immediately gives us an algorithm with runtime $\O(m^{3/2}\log U).$ Precisely, note that $\|\rho\|_{w,4}^{-1} \ge \|\rho\|_{w,2}^{-1} \ge \Omega(F_t/m^{1/2})$ by \cref{cor:rho2control}. Therefore alternately using \cref{lemma:progress} with $\delta = \frac{1}{100\|\rho\|_{w,4}}$ and \cref{lemma:centering}, we route $m^{-1/2}$ fraction of the remaining flow in each iteration in $\O(m)$ time. Therefore, after $\O(m^{1/2}\log U)$ iterations, we will have computed a flow within one unit of optimal---at this point we can round the flow to an integral flow using known techniques (see \cite{LRS13,Madry13}) and do augmenting paths computations. Note that this algorithm does not perform any weight changes, so that $w^+ = w^- = 1$ throughout. 

%% file: weight.tex
\section{Maximum Flow in Time $m^{11/8+o(1)}U^{1/4}$}
\label{sec:weight}
In this section we introduce our new weight change algorithms and main progress step. In \cref{sec:madry16weight} we explain how to perform weight change while keeping our point $(f,y,w)$ coupled. In \cref{sec:newweight} we describe our main lemma for efficiently controlling congestion via weight change. In \cref{sec:finalweight} we introduce new weight change algorithms to control the coupling of our new point $(f^\new,y^\new,w^\new)$ beyond the guarantees provided by the control on congestion. In \cref{sec:algo} we give our full algorithm for maxflow and analysis.

Throughout the section, we assume that our graph $G$ is undirected and preconditioned. We run our algorithm until $F_t = m^{1/2-\eta}$, at which point we can use augmenting paths. Therefore, we assume $F_t \ge m^{1/2-\eta}$ throughout this section. In our analysis, we assume that $\eta \le \frac12$ and $U \le \sqrt{m}$, because if $U \ge \sqrt{m}$ then our desired runtime of $m^{11/8+o(1)}U^{1/4} = m^{3/2+o(1)}$ is already known through the work of Goldberg-Rao \cite{GR98}.

\subsection{Weight Change for Maintaining Coupling}
\label{sec:madry16weight}
We defer all proofs in this section to \cref{sec:madryweightproofs}.
During our algorithms, we need to perform weight increases while maintaining the coupling of our point $(f,y,w)$ in order to increase the resistances of edges with high congestion $\rho_e$. Algorithm \ComputeWeights~computes a weight increase inducing a desired resistance increase.
\begin{lemma}[Algorithm $\ComputeWeights$]
\label{lemma:weightvsres}
Consider a point $(f,y,w)$, and let $\zeta \in \R_{\ge0}^E$ be a vector such that edge $e$ is $\zeta_e$-coupled for all edges $e \in E$. Let $r$ be the resistances as defined in \cref{eq:defr}. Let $r' \in \R_{\ge0}^E$ be a vector of desired resistance changes. We can in $\O(m)$ time compute a vector $w'=((\wp)',(\wm)') \in \R_{\ge0}^{2E}$ of weight increases satisfying
\begin{enumerate}
\item $\frac{(\wpe)'}{\upe-f_e} = \frac{(\wme)'}{\ume+f_e}$, and edge $e$ is $\zeta_e$-coupled for the point $(f,y,w+w').$
\item $(\wpe)'+(\wme)' \le 4Ur_e's_e.$
\item The resistance on edge $e$ induced by weights $w+w'$ equals $r_e+r_e'.$
\end{enumerate}
\end{lemma}
In order to control the weight change, we can bound the slacks $s_e$ for edges of high congestion. This was shown in \cite{Madry16}.
\begin{lemma}
\label{lemma:resbound}
We have that $s_e^{-1} \le \frac{\|\rho\|_{w,2}^2}{\max\{\rpe,\rme\}}.$
\end{lemma}

\subsection{Weight Change for Controlling Congestion}
\label{sec:newweight}
Here, we state our main result that we can control the congestion $\normrho_{w^\new,\infty}$ via weight change. We wish to control $\normrho_{w^\new,\infty}$ because our step size $\d \le \frac{1}{10\normrho_{w^\new,\infty}}.$ We call our algorithm to do this $\ControlCongestion(f,y,t,w)$, which returns a vector $r'$ of resistance increases and a corresponding weight increase vector $w'$.
\begin{lemma}[Algorithm \ControlCongestion]
\label{lemma:rhoinf}
Consider a $0$-coupled point $(f,y,w)$ with parameter $t$, where $F_t \ge m^{1/2-\eta}.$ Let $r$ denote the resistances as defined in \cref{eq:defr}. There is an algorithm that in $m^{1+o(1)}$ time returns a resistance increase vector $r' \in \R^E_{\ge0}$ and weight increase vector $w' \in \R^{2E}_{\ge0}$ satisfying the following. Below, $w^\new = w+w'$.
\begin{enumerate}
\item The weight change $w'$ is computed through using \cref{lemma:weightvsres} on $r'$. \label{item:weightvsres}
\item $\|r'\|_1 \le m^{6\eta+o(1)}$. \label{item:res}
\item The new congestion vector $\rho$ satisfies $\|\rho\|_{w^\new,\infty} \le m^{-\eta}\|\rho\|_{w^\new,2}.$ \label{item:cong}
\item The new congestion vector $\rho$ satisfies $s_e\max\{\rpe,\rme\} \le \frac{1}{100}m^{-3\eta}\|\rho\|_{w^\new,2}$ for all edges $e$. \label{item:technical}
\end{enumerate}
\end{lemma}
Here, we give some intuition for the conclusion of \cref{lemma:rhoinf}. Item \ref{item:weightvsres} is necessary to increase the resistances by $r'$ while maintaining coupling. Item \ref{item:res} says that we have not increased the resistances by too much. Combining this with item \ref{item:weightvsres} allows us to bound the weight change too. Item \ref{item:cong} bounds the congestion of the electric $\chi$-flow induced by weights $w^\new$. Finally, item \ref{item:technical} is a condition that enables us to perform weight reduction step in \cref{lemma:weightreduce}.

We prove \cref{lemma:rhoinf} in \cref{sec:proofrhoinf}. In the remainder of this section, we use \cref{lemma:rhoinf} to achieve an efficient maxflow algorithm.

\subsection{Weight Changes During the Progress Step}
\label{sec:finalweight}
Given the bound $\normrho_{w^\new,\infty} \le m^{-\eta}\normrho_{w^\new,2}$, we would like to set $\d = \frac{1}{100}m^\eta\normrho_{w^\new,2}$, and take a step of size $\d$ by \cref{lemma:progress}. As $\d \le \frac{1}{100\normrho_{w^\new,\infty}}$, a step of size $\d$ would maintain that our flow $f^\new$ is feasible. We cannot recenter such a step though, because $(f^\new,y^\new,w^\new)$ is $10\d^2\normrho_{w^\new,4}^2$-coupled by \cref{lemma:progress}. Here, we introduce new weight change algorithms to adjust weights \emph{after having computed} $(f^\new,y^\new,w^\new)$ that improve the coupling and allow for a centering step.

We first state our full progress step algorithm.

\begin{algorithm}[h!]
\caption{$\Progress(f,y,t,w)$. Takes in point $(f,y,w)$ that is $0$-centered for parameter $t$. Returns $(f^\new,y^\new,t^\new,w^\new)$ that is $\frac{1}{100}$-coupled.}
$(w', r') \assign \ControlCongestion(f,y,t,w), r^\new \assign r+r', w^\new \assign w+w'$ \label{line:res} \Comment{\cref{lemma:rhoinf}} \\
$\hf, \phi$ are the electric $\chi$-flows/potentials induced by $r^\new$. \label{line:electric} \\
For all $e \in E$, $\rpe \assign \frac{|\hf_e|}{\upe-f_e}, \rme \assign \frac{|\hf_e|}{\ume+f_e}.$ \label{line:rho} \\
$\d \assign \frac{1}{100}m^\eta\normrho_{w^\new,2}, f^\new \assign f+\d\hf, y^\new \assign y+\d\phi$ \Comment{\cref{lemma:progress}} \label{line:delta} \\
$w'' \assign \ReduceWeights(f^\new,w'), w^{\new_2} \assign w+w''.$ \label{line:reduce} \Comment{\cref{lemma:weightreduce}} \\
$S_1 \assign \left\{e:\max\{\rpe,\rme\} \ge m^{-2\eta}\|\rho\|_{w^\new,2} \right\}, S_2 \assign \left\{e: r_e' \ge r_e \right\}.$ \label{line:s1s2} \\
$w''' \assign \PerfectCenter(f^\new,y^\new,w^{\new_2},S_1\cup S_2), w^{\new_3} \assign w+w''+w'''.$ \label{line:perfect} \Comment{\cref{lemma:perfectweight}} \\
Return $(f^\new,y^\new,t+\d,w^{\new_3}).$
\label{algo:progress}
\end{algorithm}

\begin{algorithm}[h!]
\caption{$\PerfectCenter(f,y,w,S)$. Takes in point $(f,y,w)$ and subset $S \subseteq E$. Return $w'$ such that all edges $e \in S$ are $0$-coupled for point $(f,y,w+w')$.}
$g_e \assign (y_v-y_u) - \left(\frac{\wpe}{\upe-f_e} - \frac{\wme}{\ume+f_e} \right).$ \label{line:defg} \\
$w' = ((\wpe)', (\wme)') \assign 0.$ \Comment{Initialize weight increase weight $w'$.} \\
\For{$e \in S$}{
	\If{$g_e \ge 0$}{
		$(\wpe)' = (\upe-f_e)g_e$ \\
	}
	\Else {
		$(\wme)' = -(\ume+f_e)g_e$ \\
	}		
}
Return $w'$.
\label{algo:perfectweight}
\end{algorithm}

\begin{algorithm}[h!]
\caption{$\ReduceWeights(f,w)$. Takes in a flow $f$ and weight increases $w$. Returns $w'$ of smaller weight changes that induces the same coupling as $w$.}
$w' \assign 0$. \\
\For{$e \in E$}{
	Let $(\wpe)' \ge 0$ and $(\wme)' \ge 0$ be minimal such that $\frac{(\wpe)'}{\upe-f_e}-\frac{(\wme)'}{\ume+f_e} = \frac{\wpe}{\upe-f_e}-\frac{\wme}{\ume+f_e}.$ \label{line:comp} \\
}
Return $w'$.
\label{algo:reduceweights}
\end{algorithm}

\paragraph{Intuition for \cref{algo:progress,algo:perfectweight,algo:reduceweights}.} \cref{algo:progress} is our primary procedure for advancing along the central path by sending more flow. Its key steps are as follows. First it computes resistance changes and corresponding weight changes to control the congestion (line \ref{line:res}). This ensures that the electric flow $\hf$ computed in line \ref{line:electric} has small congestion, hence we can take a step of size $\d$ in line \ref{line:delta}. With this step computed, we reverse some of weight changes $w'$ computed in line \ref{line:res} and compute a smaller weight change vector $w''$ using \cref{algo:reduceweights} in line \ref{line:reduce}. Now, our point $(f^\new,y^\new,w^{\new_2})$ is not
yet $\frac{1}{100}$-coupled, so the algorithm computes a set of edges $S_1 \cup S_2$ to perfectly center in line \ref{line:s1s2}. Here, $S_1$ is the set of edges that are highly congested, and $S_2$ is the set of edges whose resistance is significantly decreased through the weight reduction step in line \ref{line:reduce}. Perfectly centering an edge $e$ involves performing weight change to make edge $e$ $0$-coupled, and this is done in \cref{algo:perfectweight}. The weight change $w'''$ necessary to perform perfect centering is computed in line \ref{line:perfect}. \\

We now analyze algorithm \PerfectCenter~(\cref{algo:perfectweight}).
\begin{lemma}[Algorithm \PerfectCenter]
\label{lemma:perfectweight}
Let edge $e$ be $\zeta_e$-coupled for all $e$ for point $(f,y,w)$. Algorithm $\PerfectCenter(f,y,w,S)$ in $O(|S|)$ time returns weight increases $w'$ such that
edge $e$ is $0$-coupled for point $(f,y,w+w')$ for all $e \in S$. Also, $(\wpe)' = (\wme)' = 0$ for all $e \not\in S$ and $(\wpe)'+(\wme)' \le O(U\zeta_e)$ for all $e \in S$.
\end{lemma}

\begin{proof}
Set $g_e$ as in line \ref{line:defg} of \cref{algo:perfectweight}. By assumption, $|g_e| \le \zeta_e$. $w'$ is chosen precisely so that edge $e$ is $0$-coupled for all $e \in S$. Also, $\max\{(\upe-f_e)|g_e|, (\ume+f_e)|g_e|\} \le 2U\zeta_e$ as desired.
\end{proof}

Algorithm \ReduceWeights~(\cref{algo:reduceweights}) gives a way to reverse weight changes done in algorithm \ControlCongestion. Here, we analyze the effect of doing these weight reductions.
\begin{lemma}[Analysis of \ReduceWeights]
\label{lemma:weightreduce}
Consider a $0$-coupled point $(f,y,w)$ with parameter $t$, where $F_t \ge m^{1/2-\eta}.$ After line \ref{line:reduce} in an execution of $\Progress(f,y,t,w)$ we have computed in $m^{1+o(1)}$ time a point $(f^\new,y^\new,w^{\new_2})$, weight increase vectors $w', w''$, and resistance vectors $r', r^\new$ satisfying the following.
\begin{enumerate}
\item For point $(f^\new,y^\new,w^{\new_2})$ with parameter $t+\d$, edge $e$ is $5\d^2\left(\frac{(\wpe)^\new|\rpe|^2}{\upe-f_e} + \frac{(\wme)^\new|\rme|^2}{\ume+f_e}\right)$-coupled. Here, $\rho$ is the congestion vector for point $(f,y,w^\new).$
\item $\|r'\|_1 \le m^{6\eta+o(1)}$ and $s_e\max\{\rpe,\rme\} \le \frac{1}{100}m^{-3\eta}\|\rho\|_{w^\new,2}$. \label{item:fromearlier}
\item $\|w'\|_1 \le m^{6\eta+o(1)}U^2$ and $\|w''\|_1 \le m^{4\eta+o(1)}U.$ \label{item:weightchange}
\item $\d = \frac{1}{100}m^\eta\normrho_{w^\new,2}^{-1}$.
\end{enumerate}
\end{lemma}
\begin{proof}
Because $w'$ was computed from $r'$ using \cref{lemma:weightvsres} we have that $\frac{(\wp_e)'}{\upe-f_e} = \frac{(\wme)'}{\ume+f_e}$ and $(\wpe)' + (\wme)' \le 4Ur_e's_e.$ Item \ref{item:fromearlier} follows from \cref{lemma:rhoinf}. As $s_e \le U$ we have that $\|w'\|_1 \le 4U^2\|r'\|_1 \le m^{6\eta+o(1)}U^2.$

Because $\normrho_{w^\new,\infty} \le m^{-\eta}\normrho_{w^\new,2}$ we have that for the choice $\d = \frac{1}{100}m^\eta\normrho_{w^\new,2}^{-1}$ that $\d \le \frac{1}{100\normrho_{w^\new,\infty}}.$ Now, use \cref{lemma:progress} with the choice $\d = \frac{1}{100}m^\eta\normrho_{w^\new,2}^{-1}$ to compute $f^\new = f+\d\hf$ and $y^\new$. By \cref{lemma:progress} we have that edge $e$ is $5\d^2\left(\frac{(\wpe)^\new|\rpe|^2}{\upe-f_e} + \frac{(\wme)^\new|\rme|^2}{\ume+f_e}\right)$-coupled for weights $w^\new = w+w'.$

By the definition in line \ref{line:comp} of \cref{algo:reduceweights} we have \[ \frac{(\wpe)'}{\upe-f^\new_e} - \frac{(\wme)'}{\ume+f^\new_e} = \frac{(\wpe)'}{\upe-f_e-\d\hf_e} - \frac{(\wme)'}{\ume+f_e+\d\hf_e} = \frac{(\wpe)''}{\upe-f_e-\d\hf_e} - \frac{(\wme)''}{\ume+f_e+\d\hf_e}. \] Recall that $\frac{(\wp_e)'}{\upe-f_e} = \frac{(\wme)'}{\ume+f_e}$, and \[ \frac{\upe-f_e-\d\hf_e}{\ume+f_e+\d\hf_e} = (1\pm O(\d\max\{\rpe,\rme\})) \cdot \frac{\upe-f_e}{\ume+f_e}. \] Therefore, it follows that \[ (\wpe)'' + (\wme)'' \le O(\d\max\{\rpe,\rme\}) \cdot ((w_e^+)' + (w_e^-)'). \]
By \cref{lemma:rhoinf} we have that
\begin{align*}
\|w''\|_1 &\le \sum_{e\in E} O(\d\max\{\rpe,\rme\}) \cdot ((w_e^+)' + (w_e^-)') \le O(\d) \sum_{e\in E} \max\{\rpe,\rme\}\cdot Ur_e's_e \\ &\le O(\d m^{-3\eta}\normrho_{w^\new,2}U\|r'\|_1) \le m^{4\eta+o(1)}U
\end{align*}
by our choice of $\d$ as desired.
\end{proof}

We now combine \cref{lemma:perfectweight,lemma:weightreduce} to analyze algorithm $\Progress(f,y,t,w)$.
\begin{lemma}[Algorithm $\Progress$]
\label{lemma:finalweight}
Consider a $0$-coupled point $(f,y,w)$ for parameter $t$, and $F_t \ge m^{1/2-\eta}.$ Algorithm $\Progress(f,y,t,w)$ computes in $m^{1+o(1)}$ time a point $(f^\new, y^\new, w^{\new_3})$ that is $\frac{1}{100}$-coupled for parameter $t+\d$. Additionally, $\d \ge \Omega(F_t/m^{1/2-\eta})$ and $\|w^{\new_3}\|_1-\|w\|_1 \le O(m^{5\eta}U).$
\end{lemma}
\begin{proof}
Define $f^\new, y^\new, w^\new, w^{\new_2}, w^{\new_3}, w', w'', w''', r', r^\new$ as in \cref{algo:progress}, and let $\rho$ be the congestion vector for point $(f, y, w^\new).$

The algorithm runs in $m^{1+o(1)}$ time. We can bound \[ \d = \frac{1}{100}m^\eta\normrho_{w^\new,2}^{-1} = \Omega(F_t/m^{1/2-\eta}) \] by \cref{lemma:rho2control}. We have that $\|w^{\new_3}\|_1-\|w\|_1 = \|w''\|_1 + \|w'''\|_1.$ By \cref{lemma:weightreduce} we have that $\|w''\|_1 \le m^{4\eta+o(1)}U.$ To bound $\|w'''\|_1$ we separately bound the contributions from the cases $e \in S_1$ and $e \in S_2.$ For the former case, we can use \cref{lemma:weightreduce} and \cref{lemma:perfectweight} to bound the contribution from edges $e \in S_1$ by
\begin{align*}
&O(U) \cdot \sum_{e\in S_1} 5\d^2\left(\frac{(\wpe)^\new|\rpe|^2}{\upe-f_e} + \frac{(\wme)^\new|\rme|^2}{\ume+f_e}\right) \\
&\le O(U\d^2) \sum_{e\in S_1} \left((\wpe)^\new|\rpe|^2 + (\wme)^\new|\rme|^2\right)s_e^{-1} \\
&\le O(U\d^2) \sum_{e\in S_1} \left((\wpe)^\new|\rpe|^2 + (\wme)^\new|\rme|^2\right)\frac{\normrho_{w^\new,2}^2}{\max\{\rpe,\rme\}} \\
&\le O(U\d^2m^{2\eta}\normrho_{w^\new,2}) \sum_{e\in S_1} \left((\wpe)^\new|\rpe|^2 + (\wme)^\new|\rme|^2\right) \\ 
&\le O(U\d^2m^{2\eta}\normrho_{w^\new,2}^3) \le O(Um^{4\eta}\normrho_{w^\new,2}) = O(m^{5\eta}U)
\end{align*}
where we used \cref{lemma:resbound} and \cref{cor:rho2morecontrol}. The contribution from $e \in S_2$ can be bounded by
\begin{align*}
&O(U) \cdot \sum_{e\in S_2} 5\d^2\left(\frac{(\wpe)^\new|\rpe|^2}{\upe-f_e} + \frac{(\wme)^\new|\rme|^2}{\ume+f_e}\right) \\
&\le O(U\d^2) \sum_{e\in S_2} \left((\wpe)^\new|\rpe|^2 + (\wme)^\new|\rme|^2\right)s_e^{-1} \\
&= O(U\d^2) \sum_{e\in S_2} (r_e+r_e')\hf_e^2s_e^{-1} \le O(U\d^2) \sum_{e\in S_2} r_e'\hf_e^2s_e^{-1}\\
&= O(U\d^2) \sum_{e\in S_2} r_e's_e\max\{\rpe,\rme\}^2 \le O(U\d^2m^{-3\eta}\normrho_{w^\new,2}) \sum_{e\in S_2}r_e'\max\{\rpe,\rme\} \\
&\le O(U\d^2m^{-3\eta}\normrho_{w^\new,2}\|r'\|_1\normrho_{w^\new,\infty}) \le O(U\d^2\normrho_{w^\new,2}^2m^{-4\eta}\|r'\|_1) \le m^{4\eta+o(1)}U
\end{align*}
where we have used \cref{lemma:rhoenergy}, various bounds from \cref{lemma:rhoinf}, and our choice of $\d$. Therefore, we have that
$\|w''\|_1+\|w'''\|_1 = O(m^{5\eta}U)$.

We now show that our point $(f^\new,y^\new,w^{\new_3})$ is $\frac{1}{100}$-coupled. Define $S = E \bs (S_1 \cup S_2).$ Define $r^{\new_3}$ to be the resistances induced by $(f^\new,y^\new,w^{\new_3})$. Note that for all edges $e \in S$ we have that $r^{\new_3}_e \ge \frac{1}{4}r^{\new}_e$ by our definition of $S_2$.
Then the coupling of our point $(f^\new,y^\new,w^{\new_3})$ is bounded by
\begin{align*}
&\left(\sum_{e\in S}(r^{\new_3}_e)^{-1} \left(5\d^2\left(\frac{(\wpe)^\new|\rpe|^2}{\upe-f_e} + \frac{(\wme)^\new|\rme|^2}{\ume+f_e}\right)\right)^2 \right)^\frac12 \\
&\le 10\d^2 \left(\sum_{e\in S}(r^\new_e)^{-1} \left(\frac{(\wpe)^\new|\rpe|^2}{\upe-f_e} + \frac{(\wme)^\new|\rme|^2}{\ume+f_e}\right)^2 \right)^\frac12
 \\
&\le 10\d^2 \left(\sum_{e\in S}\left(\frac{(\wpe)^\new}{(\upe-f_e)^2} + \frac{(\wme)^\new}{(\ume+f_e)^2}\right)^{-1} \left(\frac{(\wpe)^\new|\rpe|^2}{\upe-f_e} + \frac{(\wme)^\new|\rme|^2}{\ume+f_e}\right)^2 \right)^\frac12 \\
&\le 10\d^2 \left(\sum_{e\in S}\left((\wpe)^\new|\rpe|^4+(\wme)^\new|\rme|^4\right)\right)^\frac12 \\
&\le 10\d^2 \left(m^{-4\eta}\normrho_{w^\new,2}^2 \sum_{e\in S}\left((\wpe)^\new|\rpe|^2+(\wme)^\new|\rme|^2\right)\right)^\frac12 \\
&\le 10\d^2 \cdot \left(m^{-4\eta} \normrho_{w^\new,2}^4\right)^\frac12 = 10\d^2 m^{-2\eta} \normrho_{w^\new,2}^2 \le \frac{1}{100}
\end{align*}
where we have used Cauchy-Schwarz, the definition of $S_1$, and our definition of $\d$.
\end{proof}

\subsection{Full Algorithm and Analysis}
\label{sec:algo}
Here we provide our \Maxflow~algorithm (\cref{algo:maxflow}). For simplicity, we assume that our starting graph has been preconditioned as discussed in \cref{sec:precondition}.

\begin{algorithm}[H]
\caption{$\Maxflow(G)$. Takes an undirected graph $G$ with maximum capacity $U$, and which has been preconditioned. Returns the maximum $ab$ flow in $G$.}
$\eta \assign \log_m(m^{1/8-o(1)}U^{-1/4}).$ \\
$f \assign 0, y \assign 0, t \assign 0, w \assign 1$. \\
\While{$F_t \ge m^{1/2-\eta}$}{ \label{line:mainwhile}
	$(f,y,t,w) \assign \Progress(f,y,t,w)$. \Comment{Progress step.} \label{line:progress} \\
	\For{$i = 1$ to $\O(1)$} { 
		$(f,y,t,w) \assign \Center(f,y,t,w)$. \Comment{Centering step.} \label{line:centering} \\
	}
}
Round to an integral flow and use augmenting paths until done.
\label{algo:maxflow}
\end{algorithm}

\begin{proof}[Proof of \cref{thm:main}]

We first show that line \ref{line:mainwhile} executes $\O(m^{1/2-\eta})$ times. By \cref{lemma:finalweight}, we have that $\delta \ge \Omega(F_t/m^{1/2-\eta})$, hence the amount of remaining flow $F_t$ decreases by a $1 - m^{-1/2+\eta}$ factor per iteration of line \ref{line:progress}. Hence, line \ref{line:progress} and \ref{line:mainwhile} execute $\O(m^{1/2-\eta})$ times.

Now we argue that point $(f,y,w)$ before execution of line \ref{line:progress} is $0$-coupled, and point $(f,y,w)$ before execution of line \ref{line:centering} is $\frac{1}{100}$-coupled. We proceed by induction. By \cref{lemma:finalweight}, the point $(f,y,w)$ before execution of line \ref{line:centering} is $\frac{1}{100}$-coupled. After applying \cref{lemma:centering} $\O(1)$ times the point $(f,y,w)$ is now $0$-coupled.

We now show that the invariant $\|w\|_1 \le 3m$ holds throughout. By \cref{lemma:finalweight}, the total weight at the end is bounded by \[ O(m^{1/2-\eta} \cdot m^{5\eta}U) = O(m^{1/2+4\eta}U) \le m/2 \] by our choice of $\eta$. Because of our use of \cref{lemma:weightreduce} there is a temporary weight increase of $m^{6\eta+o(1)}U^2$ during each iteration of line \ref{line:progress}. By our choice of $\eta$, we can bound $m^{6\eta+o(1)}U^2 \le m/2.$ Therefore, at all times the total weight is at most $2m + m/2 + m/2 = 3m$, so $\|w\|_1 \le 3m$ holds throughout.

Finally, we show that $\Maxflow(G)$ runs in time $m^{3/2-\eta+o(1)}$. Each of line \ref{line:progress},\ref{line:centering} execute $\O(m^{1/2-\eta})$ times. By \cref{lemma:finalweight,lemma:centering} each execution requires $m^{1+o(1)}$ time. Therefore, the total runtime of these is $m^{3/2-\eta+o(1)}.$ We can round to an integral flow in $\O(m)$ time \cite{LRS13,Madry13}. The augmenting paths computation at the end also requires $O(m \cdot m^{1/2-\eta}) = O(m^{3/2-\eta})$ time.
\end{proof}

\section{Efficient Congestion Control via Energy Maximization}
\label{sec:proofrhoinf}
In this section, we prove \cref{lemma:rhoinf}. We wish to efficiently compute a weight vector $w^\new$ such that the electric $\chi$-flow associated to point $(f,y,w^\new)$ has congestion vector satisfying $\|\rho\|_{w^\new,\infty} \le m^{-\eta}\normrho_{w^\new,2}.$ As the proofs in this section involve changing the resistances and computing electric flows, it is convenient to define $\hf_r$ as the electric $\chi$-flow with resistances $r$ for all vectors $r \in \R^E_{>0}.$

If we \emph{boost} an edge $e$ with large $\rpe$ or $\rme$ by increasing the corresponding resistance $r_e$, then we significantly increase the energy of the corresponding electric flow. The following lemma shows this formally. It is proved in \cref{proofs:energyboost} and was used in previous literature \cite{CKMST11, Madry13, Madry16}. Recall that we have defined $\hf_r$ and $\hf_{r+r'}$ as the electric $\chi$-flows for resistances $r$ and $r+r'$ respectively.
\begin{lemma}
\label{lemma:energyboost}
Let resistances $r$ be defined as in \cref{eq:defr}, and let $\rho$ be the corresponding congestion vector. For any vector $r' \in \R^E_{\ge0}$ where $0 \le r_e' \le r_e$ for all edges $e \in E$, we have that
\[ \log \E_{r+r'}(\hf_{r+r'}) - \log \E_{r}(\hf_r) \ge \frac12 \sum_{e\in E} \frac{r_e's_e^2\max\{\rpe,\rme\}^2}{\normrho_{w,2}^2}. \]
\end{lemma}

The approach of \cite{Madry16} combines \cref{lemma:energyboost,lemma:weightvsres,lemma:resbound} to show that it is possible to significantly increase the electric energy without too much weight change when the electric flow is congested, i.e. $\normrho_{w,\infty}$ is high. This may require multiple rounds of weight changes, as an edge $e'$ may become congested as the result of boosting edge $e$. To more efficiently compute weight changes that give our desired congestion bound, we instead look at the problem of maximizing the increase in the logarithm of electric energy given a weight change budget of $W$.

Formalizing this, consider the electric $\chi$-flow $\hf_r$ given by the resistances $r$, which are induced by the weights $w$. Define the matrices $R = \diag(r), S = \diag(s), C = 4US,$ where $s$ is the vector of slacks. $C$ is defined so that if we had wanted to increase the resistances from $r$ to $r+r'$, then it would suffice to increase the total sum of weights from $\|w\|_1$ to $\|w\|_1 + \|Cr'\|_1$ by \cref{lemma:weightvsres}. Consequently, we consider the following problem of maximizing the energy increase under a total weight change budget of $W$:
\begin{equation}
\max_{\substack{r' \ge 0 \\ \|Cr'\|_1 \le W}} \E_{r+r'}(\hf_{r+r'}). \label{eq:origprob}
\end{equation}
We use Sion's minimax theorem as follows to rewrite the problem, where $R' = \diag(r')$.
\begin{align}
&\max_{\substack{r' \ge 0 \\ \|Cr'\|_1 \le W}} \E_{r+r'}(\hf_{r+r'}) = \max_{\substack{r' \ge 0 \\ \|Cr'\|_1 \le W}} \min_{B^Tf=\chi} f^T(R+R')f = \min_{B^Tf=\chi} \max_{\substack{r' \ge 0 \\ \|Cr'\|_1 \le W}} f^T(R+R')f \\
&= \min_{B^Tf=\chi} \|f\|_R^2 + W\|C^{-1}f^2\|_\infty = \min_{B^Tf=\chi} \|f\|_R^2 + W\|C^{-1/2}f\|_\infty^2, \label{eq:mainprobinf}
\end{align}
where $f^2$ is the entry-wise square of the vector $f$. Also, we define
\begin{equation} \label{eq:g} g(W) \defeq \max_{\substack{r' \ge 0 \\ \|Cr'\|_1 \le W}} \log \E_{r+r'}(\hf_{r+r'}) - \log \E_r(\hf_r). \end{equation}
Intuitively, $g(W)$ is the maximum (multiplicative) increase that we can cause in the electric energy with weight change budget $W$. If we have an exact solver in almost linear time for \cref{eq:mainprobinf}, then we can compute $g(W)$ in almost linear time for each $W$. Additionally, we can recover the vector $r'$ achieving this using standard techniques. Before continuing, we show that $g$ is concave by showing that energy is concave in edge resistances.
\begin{lemma}[Energy is concave in $r$]
\label{lemma:energyconcave}
For a graph $G$ with resistances $r$, let $\hf_r$ be the electric $\chi$-flow given by the resistances $r$. Then the function $\E_r(\hf_r)$ is concave in $r$ as a function on $\R^E_{\geq 0}$.
\end{lemma}
\begin{corollary}
\label{lemma:gconcave}
For diagonal matrix $C$ with positive entries on the diagonal, define $g_q: \R_{\ge0}\to \R$ as \[ g_q(W) \defeq \max_{\substack{r' \ge 0 \\ \|Cr'\|_q \le W}} \log \E_{r+r'}(\hf_{r+r'}) - \log \E_r(\hf_r) \] for some $q \in [1,\infty],$ where $\hf_{r+r'}$ and $\hf_r$ are the electric $\chi$-flows for resistances for $r+r'$ and $r$ respectively. Then $g_q$ is concave and $g_q(W)/W$ is decreasing in $W$.
\end{corollary}
We prove these in \cref{proofs:concave}. Now, we fix a parameter $d$ and compute a $W$ such that $g(W) \le dW$ assuming that an exact linear time solver exists for \cref{eq:mainprobinf}. We do this for two reasons. First, as $g(W) \le dW$ we know that the logarithm of electric energy increased by $dW$ while our total $\ell_1$ norm of weights increased by at most $W$. This will allows us to control the total weight change by bounding the electric energy. Second, we can leverage \cref{lemma:gconcave}, \cref{lemma:energyboost}, and \cref{lemma:weightvsres} to control $\|\rho\|_{w,\infty}$ where $\rho$ is the new congestion vector induced after the weight changes. Essentially, we argue that if there is an edge $e$ where $\rpe$ or $\rme$ is high, then by \cref{lemma:energyboost} that we can get a significant energy increase by increasing $r_e$. We bound the weight increase through \cref{lemma:weightvsres}, and use concavity (\cref{lemma:gconcave}) and our choice of $W$ to bound $\rpe$ and $\rme$. This gives a bound on $\normrho_{w,\infty}.$

Unfortunately, we do not have an exact solver for \cref{eq:mainprobinf}, but due to recent work of Kyng \etal~\cite{KPSW19} on optimization of smoothed $\ell_2$-$\ell_p$ flows we have a high accuracy solver for a similar problem. We have the following theorems, which show that we can find flows minimizing combinations of the $2$ and $p$ norms to high accuracy in almost linear time. Precisely, \cref{thm:l2lpp} was stated with errors $\frac{1}{\poly(m)}$ in Theorem 1.1 in \cite{KPSW19}, but the analysis shows that error $\err$ is achievable. We elaborate in \cref{sec:KPSWpoly} in the appendix.

\begin{theorem}[Theorem 1.1 in \cite{KPSW19}]
\label{thm:l2lpp}
Consider $p \in (\omega(1), (\log n)^{2/3-o(1)})$, $r \in \R^E_{\ge0}$, demand vector $d\in \R^V$, and initial solution $f_0\in\R^E$ such that all parameters are bounded by $2^{\poly(\log m)}$ and $B^Tf_0=d$. For a flow $f$, define 
\[ \val(f) \defeq \left(\sum_{e \in E} r_ef_e^2\right) + \|f\|_p^p
\enspace \text{ and } \enspace
OPT \defeq \min_{B^Tf=d} \val(f). \] 
There is an algorithm that in $m^{1+o(1)}$ time computes a flow $f$ such that $B^Tf=d$ and
\[ \val(f)-OPT \le \err(\val(f_0)-OPT) + \err. \] 
\end{theorem}
Through a standard reduction, we can also solve a ``homogeneous" smoothed $\ell_2$-$\ell_p$ flow problem.
\begin{theorem}[Homogenous smoothed $\ell_2$-$\ell_p$ solver]\label{thm:l2lp}
Consider $p \in (\omega(1), (\log n)^{2/3-o(1)})$, $r \in \R^E_{\ge0}$, demand vector $d\in \R^V$, and initial solution $f_0\in\R^E$ such that all parameters are bounded by $2^{\poly(\log m)}$ and $B^Tf_0=d$. For a flow $f$, define  \[ \val(f) \defeq \left(\sum_{e \in E} r_ef_e^2\right) + \|f\|_p^2
\enspace \text{ and } \enspace
 OPT \defeq \min_{B^Tf=d} \val(f)
 ~. 
 \]
There is an algorithm that in $m^{1+o(1)}$ time computes a flow $f$ such that $B^Tf=d$ and
\[ \val(f)-OPT \le \err(\val(f_0)-OPT) + \err. \]
\end{theorem}
We prove \cref{thm:l2lp} by a reduction to \cref{thm:l2lpp} in \cref{proofs:l2lp}. Essentially, we show that for positive convex functions $f$ and $g$ we can compute a minimizer for $f+g$ given an oracle for minimizing $f+Cg^{p/2}$ for any constant $C$.

Now, we explain how to apply the results of \cref{thm:l2lp}. First, we smooth the constraint in \cref{eq:origprob} to a $\ell_q$ constraint for $q \approx 1 + \frac{1}{\sqrt{\log m}}$ instead of an $\ell_1$ constraint. Precisely, we proceed as follows. Below, $p = \lceil\sqrt{\log m}\rceil$ and $q$ is such that $\ell_p$ and $\ell_q$ are dual and $p^{-1} + q^{-1} = 1.$
\begin{align}
&\max_{\substack{r' \ge 0 \\ \|Cr'\|_q \le W}} \E_{r+r'}(\hf_{r+r'}) = \max_{\substack{r' \ge 0 \\ \|Cr'\|_q \le W}} \min_{B^Tf=\chi} f^T(R+R')f = \min_{B^Tf=\chi} \max_{\substack{r' \ge W \\ \|Cr'\|_q \le W}} f^T(R+R')f \label{eq:copy2}\\
&= \min_{B^Tf=\chi} \|f\|_R^2 + W\|C^{-1}f^2\|_p = \min_{B^Tf=\chi} \|f\|_R^2 + W\|C^{-1/2}f\|_{2p}^2, \label{eq:mainprobw}
\end{align}
As $\|Cr'\|_1 \le \|1\|_p \|Cr'\|_q \le m^{o(1)}\|Cr'\|_q$, we can still control the total true $\ell_1$ of weight change through this objective in \cref{eq:mainprobw}. Unfortunately, we still cannot solve \cref{eq:mainprobw}, as the solver in \cref{thm:l2lp} only applies to the situation where the $\ell_p$ part of the objective ($\|f\|_p^2$) has unit coefficients. To resolve this we carefully choose the cost matrix $C$ so that the resulting optimization problem we get has unit weights on the $\ell_p$ part of the objective. This choice creates a need for weight reductions (see \cref{lemma:weightreduce}). We remark that our algorithm can be simplified if we had access to a solver allowing weights in \cref{sec:discussion}.

Before proving \cref{lemma:rhoinf}, we formally show that the results of \cref{thm:l2lp} can give us a solution to a variant of our original energy maximization problem \cref{eq:origprob}.
\begin{lemma}
\label{lemma:energymax}
Let $G$ be a graph with $m$ edges. Let $p = \sqrt{\log m}$ and let $q$ be such that $p^{-1}+q^{-1}=1.$ Let $r \in \R^E_{>0}$ be a vector of resistances. Define \[ OPT \defeq \max_{\substack{\|r'\|_q \le W \\ r' \ge 0}} \E_{r+r'}(\hf_{r+r'}). \] For all $W \ge 0$, there is an algorithm that in $m^{1+o(1)}$ time computes a vector $r' \in \R_{\ge0}^E$ such that $\|r'\|_q \le W$ and
\[ \E_{r+r'}(\hf_{r+r'}) \ge \left(1 - \err\right)OPT - \err. \]
\end{lemma}
\begin{proof}[Proof of \cref{lemma:energymax}]
For the purposes of this proof, we assume that the solver in \cref{thm:l2lp} is an exact solver. We deal with inexactness issues in \cref{sec:approximate}.

Let $R = \diag(r), R' = \diag(r').$ The manipulation in \cref{eq:copy2,eq:mainprobw} with $C = 1$ shows
\begin{equation}
\max_{\substack{\|r'\|_q \le W \\ r' \ge 0}} \E_{r+r'}(\hf_{r+r'}) = \min_{B^Tf=\chi} f^TRf + W\|f\|_{2p}^2. \label{eq:expression}
\end{equation}
Let $f^*$ be the optimum for \cref{eq:expression}, which we can compute in almost linear time by \cref{thm:l2lp}. Define $OPT = (f^*)^TR(f^*) + W\|(f^*)\|_{2p}^2.$ We can compute that 
\begin{equation} \g\left(f^TRf + W\|f\|_{2p}^2\right) = 2\left(R+W\frac{f^{2(p-1)}}{\|f^2\|_p^{p-1}}\right)f. \label{eq:opti} \end{equation}
Optimality conditions give that the gradient computed in \cref{eq:opti} at $f=f^*$ must be is perpendicular to the kernel of $B^T$, i.e. there must exist a vector $\phi$ satisfying
\begin{equation} B\phi = \left(R + W\frac{(f^*)^{2(p-1)}}{\|(f^*)^2\|_p^{p-1}}\right)f^*. \label{eq:phi2} \end{equation}

We set $r' = W\frac{(f^*)^{2(p-1)}}{\|(f^*)^2\|_p^{p-1}}$, so that $\|r'\|_q \le W.$ Also, the flow $f^*$ is electric for resistances $r+r'$ by \cref{eq:phi2}. Hence, $\E_{r+r'}(\hf_{r+r'}) = \E_{r+r'}(f^*) = OPT$ as desired.
\end{proof}

\begin{proof}[Proof of \cref{lemma:rhoinf}]
Let $p = \lceil\sqrt{\log m}\rceil$ and $q$ be such that $p^{-1}+q^{-1}=1$. For resistances $r$, let $\hf_r$ denote the electric $\chi$-flow induced. Write \[ g_q(W) \defeq \max_{\substack{\|r'\|_q \le W \\ r' \ge 0}} \log\E_{r+r'}(\hf_{r+r'})-\log \E_r(\hf_r). \] By \cref{lemma:energymax} with $d = \frac{1}{20000}m^{-6\eta}$ and $W = C_0d^{-1}\log m$ ($C_0$ is some large constant we choose later) we can compute $g_q(W)$ and the resistance changes $r'$ inducing the optimum for $g_q(W)$.

We now argue that $g_q(W)/W < d.$ Note that because $\hf_r$ is a unit flow,
\begin{align*}\E_{r+r'}(\hf_{r+r'}) &\le \hf_r^T(R+R')\hf_r = \E_r(\hf_r) + \hf_r^TR'\hf_r \le \|r'\|_1+\E_r(\hf_r) \\
&\le \|1\|_p\|r'\|_q+\E_r(\hf_r) \le m^{o(1)}W+\E_r(\hf_r).
\end{align*} Therefore, as $\E_r(\hf_r) \ge \frac{1}{m^2U^2}$ as $[\hf_r]_e \ge \frac{1}{m}$ for some edge $e$, and $r_e \ge \frac{1}{U^2}$, we have that
\[ \frac{g_q(W)}{W} \le \frac{\log \E_{r+r'}(\hf_{r+r'})-\log \E_r(\hf_r)}{W} \le \frac{\log\left(1+m^{o(1)}W/\E_r(\hf_r)\right)}{W} \le \frac{\log(1+m^{2+o(1)}U^2W)}{W} < d \]
for our choice $W = C_0d^{-1}\log m,$ and sufficiently large $C_0$ depending on $\eta \le \frac12$ and $U \le \sqrt{m}.$ Now, we have that \[ \|r'\|_1 \le \|1\|_p\|r'\|_q \le \|1\|_pW \le C_0\|1\|_pd^{-1}\log m \le m^{6\eta+o(1)}. \]

Fix an edge $e \in E$. Define the vector $r''$ as $r''_e = r_e$ and $r''_{e'} = 0$ for $e' \neq e.$ Define $W' = \|r+r'+r''\|_q-\|r+r'\|_q \le \|r''\|_q = r''_e.$ By \cref{lemma:gconcave} we have that $\frac{q_g(W+W')}{W+W'} \le \frac{g_q(W)}{W}$. Rearranging this, we have that
 \[ \frac{g_q(W+W')-g_q(W)}{W'} \le \frac{g_q(W)}{W} \le d. \] By \cref{lemma:energyboost} we have that
\begin{equation} g_q(W+W')-g_q(W) \ge \frac{r''_es_e^2\max\{\rpe,\rme\}^2}{2\normrho_{w^\new,2}^2}. \label{eq:adapt} \end{equation} Combining this with the previous equation gives us that \[ \frac{r''_es_e^2\max\{\rpe,\rme\}^2}{\normrho_{w^\new,2}^2} \le 2dW' \le 2dr''_e \le \frac{1}{10000}m^{-6\eta}r''_e. \] Cancelling $r''_e$ and rearranging gives us that $s_e\max\{\rpe,\rme\} \le \frac{1}{100}m^{-3\eta}\normrho_{w^\new,2}$ as desired.

Now, by \cref{lemma:resbound} we have that
\[ \frac{1}{100}m^{-3\eta}\normrho_{w^\new,2} \ge s_e\max\{\rpe,\rme\} \ge \frac{\max\{\rpe,\rme\}^2}{\normrho_{w^\new,2}^2}. \] Rearranging this and applying \cref{cor:rho2morecontrol} gives us \[ \max\{\rpe,\rme\} \le \frac{1}{10}m^{-3\eta/2}\normrho_{w^\new,2}^{3/2} \le m^{-\eta}\normrho_{w^\new,2} 
~.\]
\end{proof}

\section{Potential Simplifications and Open Problems}
\label{sec:discussion}
We briefly remark that if \cref{thm:l2lp} allowed for objectives of the form $f^TRf + \|Cf\|_p^2$ for diagonal matrices $C$, then much of the analysis may be simplified. Precisely, repeating the proof of \cref{lemma:rhoinf} would directly give $\|w'\|_1 \le m^{4\eta+o(1)}U$ without weight reduction. Additionally, weight reduction could not help improve this bound because $\delta\max\{\rpe,\rme\} = \Omega(1)$ in the tight case of our analysis. Also, we would not need the set $S_2$ in the proof of \cref{lemma:finalweight}, as set $S_2$ was defined strictly to analyze edges whose resistance significantly reduced after a weight reduction step. Consequently, optimizing these weighted objectives is a key open problem left by this work.

Finally, we note that if we did not perform weight reduction (\cref{lemma:weightreduce}) in our current proof, then the weight increase per iteration would be $m^{6\eta+o(1)}U^2$ by \cref{lemma:weightreduce} item \ref{item:weightchange}. This would yield a total weight increase throughout the algorithm of $m^{1/2+5\eta+o(1)}U^2$ and a runtime of $m^{7/5+o(1)}U^{2/5}$ after choosing $m^\eta = m^{1/10-o(1)}U^{-2/5}.$ Interestingly, this still improves over the algorithms of \cite{Madry13,Madry16} for small values of $U$.

%% file: appprelim.tex
\section{Additional Preliminaries}
\label{sec:appprelim}
\subsection{Convex Optimization Preliminaries}
\label{sec:opt}

In this section, we state some preliminaries for convex optimization. These will be used in \cref{sec:proofs,sec:approximate,sec:numerical}. We assume all functions in this section to be convex. We also work in the $\ell_2$ norm exclusively. Proofs for the results stated can be found in \cite{Nes98}.

\paragraph{Matrices and norms.} We say that a $m\times m$ matrix $M$ is positive semidefinite if $x^TMx \ge 0$ for all $x \in \R^m.$ We say that $M$ is positive definite if $x^TMx > 0$ for all nonzero $x \in \R^m$. For $m \times m$ matrices $A, B$ we write $A \se B$ if $A-B$ is positive semidefinite, and $A \succ B$ is $A-B$ is positive definite. For $m \times m$ positive semidefinite matrix $M$ and vector $x \in \R^m$ we define $\|x\|_M = \sqrt{x^TMx}.$ For $m \times m$ positive semidefinite matrices $M_1, M_2$ and $C > 0$ we say that $M_1 \approx_C M_2$ if $\frac{1}{C} x^TM_1x \le x^TM_2x \le Cx^TM_1x$ for all $x \in \R^m.$

\paragraph{Lipschitz functions.} Here we define what it means for a function $f$ to be Lipschitz and provide a lemma showing its equivalence to a bound on the norm of the gradient.

\begin{definition}[Lipschitz Function]
Let $f: \R^n \to \R$ be a function, and let $\X \subseteq \R^n$ be an open convex set. We say that $f$ is $L_1$-Lipschitz on $\X$ (in the $\ell_2$ norm) if for all $x, y \in \X$ we have that $|f(x) - f(y)| \le L_1\|x-y\|_2.$
\end{definition}

\begin{lemma}[Gradient Characterization of Lipschitz Function]
\label{lemma:gradlip}
Let $f: \R^n \to \R$ be a differentiable function, and let $\X \subseteq \R^n$ be an open convex set. Then $f$ is $L_1$-Lipschitz on $\X$ if and only if for all $x \in \X$ we have that $\|\g f(x)\|_2 \le L_1.$
\end{lemma}

\paragraph{Smoothness and strong convexity.} We define what it means for a function $f$ to be convex, smooth, and strongly convex. We say that a function $f$ is convex on $\X$ if for all $x, y \in \X$ and $0 \le t \le 1$ that $f(tx+(1-t)y) \le tf(x)+(1-t)f(y).$ We say that $f$ is $L_2$-smooth on $\X$ if $\|\g f(x)-\g f(y)\|_2 \le L_2\|x-y\|_2$ for all $x, y \in \X$. We say that $f$ is $\mu$-strongly convex on $\X$ if for all $x, y \in \X$ and $0 \le t \le 1$ that
\[ f(tx+(1-t)y) \le tf(x)+(1-t)f(y)-t(1-t) \cdot \frac{\mu}{2}\|x-y\|_2^2. \]

\begin{lemma}
Let $f:\R^n \to \R$ be a differentiable function, and let $\X \subseteq \R^n$ be an open convex set. Then $f$ is $\mu$-strongly convex on $\X$ if and only if for all $x, y \in \X$ we have that \[ f(y) \ge f(x) + \g f(x)^T(y-x) + \frac{\mu}{2}\|y-x\|_2^2. \] Also, $f$ is $L_2$-smooth on $\X$ if and only if for all $x, y \in \X$ we have that \[ f(y) \le f(x) + \g f(x)^T(y-x) + \frac{L_2}{2}\|y-x\|_2^2. \]
\end{lemma}
We can equivalently view smoothness and strong convexity as spectral bounds on the Hessian of $f$.
\begin{lemma}
\label{lemma:equiv}
Let $f:\R^n \to \R$ be a twice differentiable function, and let $\X \subseteq \R^n$ be an open convex set. Then $f$ is $\mu$-strongly convex on a convex set $\X$ if and only if $\g^2 f(x) \se \mu I$ for all $x \in \X$. $f$ is $L_2$-smooth on $\X$ if and only if $\g^2 f(x) \pe L_2 I$ for all $x \in \X$.
\end{lemma}
Smoothness allows us to relate function error and the norm of the gradient.
\begin{lemma}
\label{lemma:smoothgrad}
Let $\X \subseteq \R^n$ be an open convex set, and let $f:\R^n \to \R$ be $L_2$-smooth on $\X.$ Define $x^* = \argmin_{x \in \R^n} f(x)$, and assume that $x^*$ exists and $x^* \in \X.$ Then for all $x \in X$ we have that
\[ \|\g f(x)\|_2^2 \le 2L_2(f(x)-f(x^*)). \]
\end{lemma}
Strong convexity allows us to relate function error and distance to the optimal point.
\begin{lemma}
\label{lemma:strong}
Let $\X \subseteq \R^n$ be an open convex set, and let $f:\R^n \to \R$ be $\mu$-strongly convex on $\X.$ Define $x^* = \argmin_{x \in \R^n} f(x)$, and assume that $x^*$ exists and $x^* \in \X.$ Then for all $x \in \X$ we have that
\[ \|x-x^*\|_2^2 \le \frac{2(f(x)-f(x^*))}{\mu}. \]
\end{lemma}

\subsection{More on Electric Flow and Laplacian Systems}
\label{sec:lap}
Laplacian systems can be solved to high accuracy in nearly linear time \cite{ST04, KMP10, KMP11, KOSZ13, CKMPPRX14, KLPSS16, KS16}.
\begin{theorem} 
\label{thm:lap}
Let $G$ be a graph with $n$ vertices and $m$ edges. Let $r \in \R_{>0}^E$ denote edge resistances. For any demand vector $d$ and $\eps>0$ there is an algorithm which computes in $\O(m \log \eps^{-1})$ time \emph{potentials} $\phi$ such that $\|\phi - \phi^*\|_L \le \eps\|\phi^*\|_L$, where $L$ is the Laplacian of $G$, and $\phi^* = L^\dagger d$ are the true potentials determined by the resistances $r$. An approximate energy minimizing flow $\hf$ can be computed in $\O(m)$ time.
\end{theorem}

Electric flows are \emph{given by potentials}. This formula is called \emph{Ohm's law} --- for every edge $e = (u,v)$ we have that $\hf_e = \frac{\phi_v-\phi_u}{r_e}.$ This gives us an alternate formula for electric energy:
 \[ \E_r(\hf) = \sum_{e\in E} r_e\hf_e^2 = \sum_{e\in E} \frac{(\phi_u-\phi_v)^2}{r_e}.\] This formula gives us a dual characterization of the energy of the optimal electric flow in terms of potentials. Further, as the following standard lemma shows, we can use any $\phi'$ with $d^T\phi' = 1$ to lower-bound energy (see, e.g. \cite{Madry13} Lemma 2.1, for proof).
\begin{lemma}
\label{lemma:energydual}
For any graph $G$ with edge resistances $r_e$, we have that
\begin{equation} \label{eq:phi} \frac{1}{\E_r(\hf)} = \min_{d^T\phi'=1} \sum_{e=(u,v)\in E} \frac{(\phi'_u-\phi'_v)^2}{r_e}, \end{equation} where $\hf$ is the electric $d$-flow given by the resistances $r$. Further, if $\phi = L^\dagger d$ are the potentials induced by the optimal electric flow, then the minimizer of \cref{eq:phi} is given by $\phi' = \phi/\E_r(\hf).$
\end{lemma}

%% file: proofs.tex
\section{Missing proofs}
\label{sec:proofs}
\subsection{Proof of \cref{lemma:rhoenergy}}
\label{proofs:rhoenergy}
\begin{proof}
We have that
\[ \wp_e |\rp_e|^2 + \wm_e |\rm_e|^2 = \wp_e \cdot \frac{\hf_e^2}{(\up_e-f_e)^2} + \wm_e \cdot \frac{\hf_e^2}{(\um_e+f_e)^2} = r_e\hf_e^2 \] by the definition of $r_e$ (see \cref{eq:defr}). Now we have that
\[ \normrho_{w,2}^2 = \sum_{e\in E} \left(\wp_e |\rp_e|^2 + \wm_e |\rm_e|^2 \right) = \sum_{e\in E} r_e\hf_e^2 = \E_r(\hf). \]
\end{proof}
\subsection{Proof of \cref{lemma:progress}}
\label{proofs:progress}
First, we state a Taylor approximation result from \cite{Madry16} that will be useful for our analysis.
\begin{lemma}
\label{lemma:taylor}
If, $u_1,u_2,w_1,w_2 > 0$,  $u \defeq \min\{u_1, u_2\}$, and $x \in \R$ satisfies $|x| \le u/4$, then 
\[ \left|\left(\frac{w_1}{u_1-x}+\frac{w_2}{u_2+x}\right) - \left(\frac{w_1}{u_1}-\frac{w_2}{u_2} + \left(\frac{w_1}{u_1^2}+\frac{w_2}{u_2^2}\right)x\right)\right| \le \left(\frac{5w_1}{u_1^3}+\frac{5w_2}{u_2^3}\right)x^2. \]
\end{lemma}
\begin{proof}[Proof of \cref{lemma:progress}]
Let $f'$ and $y'$ be the desired changes in $f$ and $y$ respectively. To maintain good coupling, we want for \cref{eq:centralpath} to still hold after the update up to second order terms. Precisely, we want
\[ B^Tf' = \delta\chi \text{ and } By' = \g^2 V(f)f' \approx \g V(f+f')-\g V(f). \]
Set $R = \g^2 V(f)$ and $L = B^TR^{-1}B$, so that $L$ is a Laplacian. Solving the above system, we get that $y' = \d L^{-1}\chi$ and $f' = R^{-1}By' = \d\hf$, i.e. $\delta$ times the electric flow with resistances $r$. By \cref{thm:lap}, we can compute $f'$ and $y'$ in $\O(m)$ time.

Set $f^\new = f+f'$ and $y^\new = y+y'$. We first compute the coupling of an edge $e = (u,v)$. Define $1_e \in \R^E$ to be the indicator vector of edge $e$. Then the coupling of edge $e$ is
\begin{align*} &\left|1_e^T\left(B(y+y')-\g V(f+f')\right)\right| = \left|1_e^T\left(By'-\left(\g V(f+f')-\g V(f)\right)\right)\right| \\
			   &= \left|\left(\frac{\wpe}{\upe-f_e}-\frac{\wme}{\ume+f_e}+\left(\frac{\wpe}{(\upe-f_e)^2}+\frac{\wme}{(\ume+f_e)^2}\right)f_e'\right)-\left(\frac{\wpe}{\upe-f_e-f_e'}-\frac{\wme}{\ume+f_e+f_e'}\right)\right| \\
			   &\le 5\left(\frac{\wpe(f_e')^2}{(\upe-f_e)^3} + \frac{\wme(f_e')^2}{(\ume+f_e)^3}\right) = 5\d^2\left(\frac{\wpe\hf_e^2}{(\upe-f_e)^3} + \frac{\wme\hf_e^2}{(\ume+f_e)^3}\right) = 5\d^2\left(\frac{\wpe|\rpe|^2}{\upe-f_e} + \frac{\wme|\rme|^2}{\ume+f_e}\right).
\end{align*}
where we have used that $(f,y)$ is $0$-coupled, \cref{lemma:taylor}, and the definition $f' = \d\hf$. This completes the proof of the second point in \cref{lemma:progress}.

We now bound how coupled $(f^\new,y^\new)$ is. We can compute from the above and Cauchy-Schwarz that
\begin{align*}&\|B(y+y')-\g V(f+f')\|_{R^{-1}} \\ &\le \left(\sum_{e \in E} \left(\frac{\wpe}{(\upe-f_e)^2}+\frac{\wme}{(\ume+f_e)^2}\right)^{-1} \left(5\d^2\left(\frac{\wpe|\rpe|^2}{\upe-f_e} + \frac{\wme|\rme|^2}{\ume+f_e}\right)\right)^2\right)^\frac12 \\
		&\le 5\d^2 \left(\sum_{e \in E} \wpe|\rpe|^4 + \wme|\rme|^4 \right)^\frac12 = 5\d^2 \normrho_{w,4}^2.
\end{align*}
Now, we have that $R' \approx_2 R$, where $R'$ are the new resistances induced by $f^\new.$ Therefore, we have that 
\[ \|B(y+y')-\g V(f+f')\|_{{R'}^{-1}} \le 2\|B(y+y')-\g V(f+f')\|_{R^{-1}} \le 10\d^2 \normrho_{w,4}^2. \] This completes the proof of the third point in \cref{lemma:progress}.
\end{proof}

\subsection{Proof of \cref{lemma:centering}}
\label{proofs:centering}
\begin{proof}
Let $f'$ and $y'$ be the desired changes in $f$ and $y$ respectively. Define $g = By-\g V(f)$ as in \cref{eq:gg}. We want for
\[ B^Tf' = 0 \text{ and } -g = By'-\g^2 V(f)f' \approx By'-(\g V(f+f')-\g V(f)). \] Define $R = \g^2 V(f)$ and $L = B^TR^{-1}B.$ Solving the previous system gives us that \[ y' = -L^{-1}B^TR^{-1}g \text{ and } f' = R^{-1}By'+R^{-1}g = R^{-1/2}P(R^{-1/2}g) \] for $P = (I-R^{-1/2}BL^{-1}B^TR^{-1/2})$, an orthogonal projection matrix. By \cref{thm:lap}, we can compute both $y'$ and $f'$ in $\O(m)$ time.

Set $f^\new = f+f'$ and $y^\new = y+y'$. We claim that $(f+f',y+y')$ is $10\gamma^2$-coupled. Because $(f,y)$ is $\gamma$-coupled, we know that $\|R^{-1/2}g\|_2^2 \le \gamma.$ By the same computation as in the proof in \cref{proofs:progress}, our construction of $f'$ and $y'$, and \cref{lemma:taylor} we have that
\begin{align*} &\|B(y+y')-\g V(f+f')\|_{R^{-1}} = \|\g V(f+f')-\g V(f) - \g^2 V(f)\|_{R^{-1}} \\
				&\le 5\left(\sum_{e\in E} \left(\frac{\wp_e}{(\up_e-f_e)^2}+\frac{\wm_e}{(\um_e+f_e)^2}\right)^{-1} \left(\left(\frac{\wp_e}{(\up_e-f_e)^3}+\frac{\wm_e}{(\um_e+f_e)^3}\right)(f_e')^2\right)^2 \right)^{1/2} \\
				&= 5\left(\sum_{e\in E} \left(\frac{\wp_e}{(\up_e-f_e)^2}+\frac{\wm_e}{(\um_e+f_e)^2}\right)^{-1} \left(\frac{\wp_e}{(\up_e-f_e)^3}+\frac{\wm_e}{(\um_e+f_e)^3}\right)^2(f_e')^4 \right)^{1/2} \\
				&\le 5\left(\sum_{e\in E} \left(\frac{\wp_e}{(\up_e-f_e)^2}+\frac{\wm_e}{(\um_e+f_e)^2}\right)^{-1} \left(\frac{(\wp_e)^{3/2}}{(\up_e-f_e)^3}+\frac{(\wm_e)^{3/2}}{(\um_e+f_e)^3}\right)^2(f_e')^4 \right)^{1/2} \\
				&\le 5\left(\sum_{e\in E} \left(\frac{\wp_e}{(\up_e-f_e)^2}+\frac{\wm_e}{(\um_e+f_e)^2}\right)^2(f_e')^4 \right)^{1/2} \\
				&= 5\|R^{1/2}f'\|_4^2 \le 5\|R^{1/2}f'\|_2^2 = 5\|P(R^{-1/2}g)\|_2^2 \le 5\|R^{-1/2}g\|_2^2 \le 5\gamma^2.
\end{align*}
Finally, because $\|R^{1/2}f\|_\infty \le \|R^{1/2}f\|_2 = \|P(R^{-1/2}g)\|_2 \le \|R^{-1/2}g\|_2 \le \gamma$, we have that $R' \approx_2 R$ where $R'$ are the new resistances induced by $f+f'.$ Therefore, \[ \|B(y+y')-\g V(f+f')\|_{{R'}^{-1}} \le 2\|B(y+y')-\g V(f+f')\|_{R^{-1}} \le 10\gamma^2 ~. \]
\end{proof}

\subsection{Reduction to undirected maxflow}
\label{proofs:undirreduc}
Here, we explain how to reduce directed maxflow to undirected maxflow \cite{Lin09,Madry13}. Let $G$ be a directed maxflow instance with single source $a$, sink $b$, and maximum edge capacity $U$. Assume that the maximum flow in $G$ has value $t^*.$

Construct $G'$ as follows. For directed edge $e = (u,v) \in E(G)$ of capacity $c_e$, add undirected edges $(a,v), (v,u), (b,u)$, each of capacity $c_e$, to graph $G'$. Consider the flow $f'$ in $G'$ given by sending $c$ units of flow along the path $a\to v \to u \to b$ for all edges $(u, v) \in E(G).$ This flow routes exactly $\sum_{e \in E(G)} c_e$ units. The residual graph $\hat{G'}(f')$ in $G'$ induced by $f$ contains a directed edge of capacity $2c_e$ for each $e \in E(G)$, and directed edges to and from $a$ and $b$ of capacity $2c_e.$ The edges towards $a$ and from $b$ cannot contribute to any flow, hence the maximum flow in $\hat{G'}(f')$ has value $2t^*$. Therefore, the maximum flow in $G'$ is $2t^* + \sum_{e \in E(G)} c_e$. In this way, we can compute the maxflow in $G$ from the maxflow in $G'$ (see \cite{Lin09,Madry13}). Also, $G'$ has $O(m)$ edges as desired and maximum capacity $U$ as desired.

\subsection{Proof of \cref{lemma:rho2control}}
\label{proofs:rho2control}
We follow the approach of \cite{Madry16}.
\paragraph{Bounding $|\chi^Ty|$.} We show the following.
\begin{lemma}
\label{lemma:chiy}
Consider a preconditioned graph $G$. Let point $(f,y,w)$ be $\frac{1}{100}$-coupled for parameter $t$. Then we have that $|\chi^Ty| \le \frac{4\|w\|_1}{F_t}.$
\end{lemma}
\begin{proof}
Let $f$ be a flow in $G$. Consider a $F_t\chi$-flow $f'$ satisfying $\um+f \le f' \le \up-f$. Let $R = \g^2 V(f)$. We first show that $\chi^Ty \le \frac{4\|w\|_1}{F_t}.$ Define $\psi = R^{-1/2}(By-\g V(f))$, so that $\|\psi\|_2 \le \frac{1}{100}$. Now
\begin{align}
F_t\chi^Ty &= (B^Tf')^Ty = f'^TBy = f'^T\g V(f) + (R^{1/2}f')^T\psi \\
		 &= \sum_{e\in E} \left(\frac{\wpe}{\upe-f_e} - \frac{\wme}{\ume+f_e} \right)f_e' + \left(\frac{\wpe}{(\upe-f_e)^2}+\frac{\wme}{(\ume+f_e)^2}\right)^\frac12f_e'\psi_e. \label{eq:longsum}
\end{align}
We now bound the contribution of edge $e$ to \cref{eq:longsum}. WLOG we assume $f_e' \ge 0$. Using $f_e' \le \upe-f_e$ and $\psi_e \le \|\psi\|_2 \le \frac{1}{100}$ we bound
\begin{align*}
&\left(\frac{\wpe}{\upe-f_e} - \frac{\wme}{\ume+f_e} \right)f_e' + \left(\frac{\wpe}{(\upe-f_e)^2}+\frac{\wme}{(\ume+f_e)^2}\right)^\frac12f_e'\psi_e \\
&\le \left(\frac{\wpe}{\upe-f_e} - \frac{\wme}{\ume+f_e} \right)f_e' + \left(\frac{\wpe}{\upe-f_e} + \frac{\wme}{\ume+f_e} \right)f_e'|\psi_e| \\
&\le \wpe(1+|\psi_e|) \le 2\wpe.
\end{align*}
Hence we can bound the quantity in \cref{eq:longsum} by $2\|w\|_1$ as desired.

Now we show $\chi^Ty \ge -\frac{4\|w\|_1}{F_t}.$ By symmetry, throughout the algorithm, all the preconditioning edges will have the same weights and flows. For any preconditioning edge $e$, we have that $f_e \ge -U.$ Indeed, if all the $m$ preconditioning edges have $f_e < -U$, then the remainder of the graph must support an $ab$-flow of value greater than $mU$, a contradiction. Now, we can write
$By = \g V(f) + R^{1/2}\psi.$ For a preconditioning edge $e$ this gives us that
\begin{align*}
y_b - y_a &= \frac{\wpe}{\upe-f_e}-\frac{\wme}{\ume+f_e}+\left(\frac{\wpe}{(\upe-f_e)^2}-\frac{\wme}{(\ume+f_e)^2}\right)^\frac12\psi_e \\
		  &\ge -\frac{\wme}{\ume+f_e}-\left(\frac{\wpe}{\upe-f_e}+\frac{\wme}{\ume+f_e}\right)|\psi_e| \\
		  &\ge -\frac{\wme}{U}-\frac{\wpe+\wme}{U}\psi_e \ge -\frac{4}{3}\frac{\wme+\wpe}{U} \ge -\frac{4}{3}\frac{\|w\|_1}{mU}
\end{align*}
where we have used $|\psi_e| \le \frac{1}{100}$ for all $e$ and that all $m$ preconditioning edges have the same weights. Now, we have
\[ \chi^Ty = y_b-y_a \ge -\frac{4}{3}\frac{\|w\|_1}{mU} \ge -\frac{4\|w\|_1}{F_t} \] where we have used $F_t \le 3mU.$
\end{proof}

\begin{proof}[Proof of \cref{lemma:rho2control}] We now use \cref{lemma:chiy} to complete the proof of \cref{lemma:rho2control}.
\paragraph{Controlling congestion of preconditioning edges.} We will show that for each of the $m$ preconditioning edges $e$ that $\min(\up_e-f_e,\um_e+f_e) \ge \frac{F_t}{7\|w\|_1}$ Recall that for each of the $m$ preconditioning edges that $\um_e = \up_e = 2U.$ By symmetry, all preconditioning edges $e$ have the same weights $\wpe$ and $\wme$.

Without loss of generality assume that $\up_e-f_e \le \um_e+f_e$, and the other case is similar. Assume $\up_e-f_e \le \frac{F_t}{7\|w\|_1}$ for contradiction. Then we have that $f_e \ge \up_e-\frac{F_t}{7\|w\|_1} \ge 2U-\frac{3mU}{14m} \ge \frac{3}{2}U.$ Then, $\um_e+f_e \ge \frac{7}{2}U.$

Finally, we have from the above and the fact that $(f,y)$ is $\frac{1}{100}$-coupled that
\begin{align*}
|\chi^Ty| &= |y_t-y_s| \ge \frac{\wpe}{\up_e-f_e}-\frac{\wme}{\um_e+f_e}-\frac{1}{100}r_e^{-1/2} \\
		  &\ge \frac{\wpe}{\up_e-f_e}-\frac{\wme}{\um_e+f_e}-\frac{1}{100}\left(\frac{\wpe}{\up_e-f_e}+\frac{\wme}{\um_e+f_e}\right) \\
		  &\ge \frac{\frac{99}{100}\wpe}{\upe-f_e}-\frac{\frac{101}{100}\wme}{\ume+f_e} 
		  \ge \frac{6\|w\|_1}{F_t} - \frac{3\|w\|_1}{7mU} > \frac{4\|w\|_1}{F_t},
\end{align*}
a contradiction. Here, we have used $F_t \le 3mU$.

\paragraph{Finishing the bound.} First we bound from above that \[ r_e \le \frac{\wpe}{(\upe-f_e)^2}+\frac{\wme}{(\ume+f_e)^2} \le \frac{49\|w\|_1^2}{F_t^2}(\wpe+\wme) \le \frac{50\|w\|_1^3}{mF_t^2} \] as $\wpe$ and $\wme$ are equal for all preconditioning edges $e$.
The energy of the electric flow is the minimum energy flow. 
Consider the flow given by sending $1/m$ units of flow across each of the $m$ preconditioning edges. Above, we have shown that for each preconditioning edge $e$ that $r_e \le \frac{50\|w\|_1^3}{mF_t^2}.$ Therefore, we have that the energy $\E_r(\hf)$ induced by resistances $r$ is at most $m \cdot \frac{50\|w\|_1^3}{mF_t^2} \cdot \frac{1}{m^2} = \frac{50\|w\|_1^3}{m^2F_t^2}.$ Finally, we have that $\|\rho\|_{w,2}^2 = \E_r(\hf) \le \frac{50\|w\|_1^3}{m^2F_t^2}$ as desired.
\end{proof}

\subsection{Proof of \cref{lemma:energyconcave} and \cref{lemma:gconcave}}
\label{proofs:concave}
\begin{proof}[Proof of \cref{lemma:energyconcave}]
Recall that $\hf_r = R^{-1} B L^\dagger d$ where $L = B^T R^{-1} B$. Consequently,
\[
\E_r(\hf_r) = \hf_r^T R \hf_r 
= d^T L^\dagger B^T R^{-1} R R^{-1} B L^\dagger d
= d^T L^\dagger d ~.
\]
For the following computation, we write it as if $L$ is invertible. Indeed, the kernel of $L$ is just the vector $1$, so that we can restrict to the orthogonal complement of the vector $1$ and perform the computations in that space. For all $e \in E$ we have
\begin{align*}
\frac{\partial}{\partial r_e}
\E_r(\hf_r)
&= -d^T L^\dagger \left(\frac{\partial}{\partial r_e} L\right) L^\dagger d
=  -d^T L^\dagger \left(- B^T 1_e 1_e^T B r_e^{-2} \right) L^\dagger d
= (1_e^T R^{-1} B L^\dagger d)^2
= [\hf_r]_e^2,
\end{align*}
where $[\hf_r]_e$ is the entry of $\hf_r$ corresponding to edge $e$.
Consequently, for all $e_1, e_2 \in E$ we have
\begin{align*}
\frac{\partial}{\partial r_{e_1}} \frac{\partial}{\partial r_{e_2}} \E_r(\hf_r)
&= \frac{\partial}{\partial r_{e_2}} [\hf_r]_{e_1}^2
= 2[\hf_r]_{e_1}  \left(-1_{e_1 = e_2} r_{e_2}^{-1} [\hf_r]_{e_2} + 
1_{e_1}^T R^{-1} B L^\dagger \left(B^T 1_{e_2} 1_{e_2}^T B r_{e_2}^{-2} \right) L^\dagger d
\right) \\
&= -2 [\hf_r]_{e_1}  [\hf_r]_{e_2} 
 \left(1_{e_1 = e_2} r_{e_2}^{-1/2} - 
[R^{-1} B L^\dagger B^T R^{-1}]_{e_1, e_2} \right)
~.
\end{align*}
where we used $1_{e_1 = e_2}$ to denote the indicator for the event that $e_1 = e_2$. Now letting $F \in \R^{m \times m}$ denote the diagonal matrix with $F_{e,e} = [\hf_r]_e$ for all $e \in E$ we have that 
\[
\nabla^2 \E_r(\hf_r)
= -2F R^{-1/2} \left(I - R^{-1/2} B L^\dagger B^T R^{-1/2}\right) R^{-1/2} F
\]
However, since $R^{-1/2} B L^\dagger B^T R^{-1/2}$ is an orthogonal projection matrix all eigenvalues of $I - R^{-1/2} B L^\dagger B^T R^{-1/2}$ are $0$ and $1$ and therefore $\nabla^2 \E_r(\hf_r)$ is negative definite as desired.
\end{proof}

\begin{proof}[Proof of \cref{lemma:gconcave}]
Define $v(r') = \log \E_{r+r'}(\hf_{r+r'}) - \log \E_r(\hf_r)$. By \cref{lemma:energyconcave}, we have that $v$ is concave in $r'$, as the logarithm of a nonnegative concave function is concave.

Take $0 \le \lambda \le 1.$ Consider $W_1, W_2 \ge 0$ and let \[ r_i = \argmax_{\substack{r' \ge 0 \\ \|Cr'\|_q \le W_i}} v(r'). \] In this way, $\|Cr_i\|_q \le W_i$ and $v(r_i) = g_q(W_i).$ Note that
\begin{align*} g_q(\lambda W_1 + (1-\lambda)W_2) &= \max_{\substack{r' \ge 0 \\ \|Cr'\|_q \le \lambda W_1 + (1-\lambda)W_2}} v(r') \ge v(\lambda r_1 + (1-\lambda)r_2) \\ &\ge \lambda v(r_1) + (1-\lambda)v(r_2) = \lambda g_q(W_1) + (1-\lambda)g_q(W_2) \end{align*} as desired. Here, we have used that \[ \|C(\lambda r_1 + (1-\lambda)r_2)\|_q \le \lambda\|Cr_1\|_q + (1-\lambda)\|Cr_2\|_q \le \lambda W_1 + (1-\lambda)W_2. \] To deduce that $g_q(W)/W$ is decreasing, it suffices to note that $g_q(0) = 0$, $g_q$ is increasing, and $g_q$ is concave.
\end{proof}

\subsection{Proof of \cref{thm:l2lp}}
\label{proofs:l2lp}
We first show a generic reduction lemma.
\begin{lemma}
\label{lemma:reduc}
Let $f, g:\R^n \to \R$ be functions such that $f$ is $\mu_f$ strongly convex. Let $T_1 > 0$ be a constant such that $f(x) \ge T_1$ and $g(x) \ge T_1$ for all $x \in \R^n.$ Let $h: \R_{\ge0} \to \R_{\ge0}$ be a monotonically increasing convex function. Let $C_0$ be a positive constant, and define \[ x^* = \min_{x \in \R^n} f(x) + C_0 \cdot g(x). \] Let $\X \subseteq \R^n$ be a region such that $x^* \in \X$ and that $f(x) \le T_2$ and $g(x) \le T_2$ for all $x \in \X.$ Define $H_1 = \min\{h(T_1), h'(T_1)\}$ and $H_2 = \max\{h(T_2), h'(T_2)\}.$ Assume that $g$ is $L_g$-Lipschitz in the $\ell_2$ norm on $\X$, and that $h$ and $h'$ are $L_h$-Lipschitz on $[T_1,T_2].$

Define \[ x_C = \argmin_{x \in \R^n} f(x) + C \cdot h(g(x)), \] and let $Z_1$ and $Z_2$ be constants such that
\begin{equation} Z_1 \cdot h'(g(x_{Z_1})) < C_0 \text{ and } Z_2 \cdot h'(g(x_{Z_2})) > C_0. \label{eq:binsearchbound} \end{equation}
Say that for all $C \in [Z_1,Z_2]$ there is an oracle running in time $\T(\eps)$ that computes $y_C \in \R^n$ satisfying
\begin{align}
&f(y_C) + C \cdot h(g(y_C)) \le f(x_C) + C \cdot h(g(x_C)) + \eps \text{ and } \label{eq:thing1} \\
&\|\g f(y_C) + C \cdot h'(g(y_C)) \cdot \g g(y_C)\|_2 \le \eps \label{eq:thing2}.
\end{align}
Then there is an algorithm that with inputs $f, g, h, C_0, \X, \mu_f, L_g, L_h, Z_1, Z_2, H_1, H_2, \eps_1$ satisfying \[ \min\{L_g,L_h,Z_2,H_2\} \ge 1 \text{ and } \max\{\mu_f,Z_1,H_1,\eps_1\} \le 1 \] runs in time
\begin{equation}  O\left(\T\left(\frac{\mu_f\eps_1^2}{100Z_2^2L_g^4L_h^2}\right) \log \frac{H_2Z_2L_gL_h}{\mu_f\eps_1} \right) \label{eq:timebound} \end{equation} and computes a $y \in \R^n$ satisfying
\[ \|\g f(y) + C_0 \g g(y)\|_2 \le \eps_1. \]
\end{lemma}
\begin{proof}
We first overview the proof, then show the necessary claims in the following paragraphs.
\paragraph{Overview.} As in the statement of \cref{lemma:reduc}, define $x_C = \argmin_{x \in \R^n} f(x) + C \cdot h(g(x))$, and define the function $w:\R_{\ge0} \to \R$ by $w(C) = C \cdot h'(g(x_C)).$ By optimality conditions, we have that \[ \g f(x_C) + C \cdot h'(g(x_C)) \g g(x_C) = f(x_C) + w(C) \g g(x_C) = 0. \] Also, we have that $\g f(x^*) + C_0 \g g(x^*) = 0.$ Thus, if we find a $C^*$ with $w(C^*) = C_0$, then $\g f(x_{C^*}) + C_0 \g g(x_{C^*}) = 0$, hence $x_{C^*} = x^*$. To compute $C^*$, we show that $w(C)$ is increasing and then binary search to compute $C^*$. Finally, we analyze the guarantees of the algorithm, using that our calls to the oracle only return approximate minimizers $y_C$.

\paragraph{Showing that $w(C)$ is increasing.} We now show that $w$ is increasing in $C$. To show this, we first compute $\frac{dx_C}{dC}.$ For simplicity, define $y_C = \frac{dx_C}{dC}.$ By optimality conditions, we can write $\g f(x_C) + C \cdot h'(g(x_C))\g g(x_C) = 0.$ Differentiating this with respect to $C$ and solving for $y_C$ gives
\begin{align*}
&\g^2 f(x_C)y_C + h'(g(x_C))\g g(x_C) + C \cdot h''(g(x_C))\g g(x_C)\g g(x_C)^T y_C + C \cdot h'(g(x_C)) \g^2g(x_C)y_C = 0 
\end{align*}
Letting $H_C \defeq \g^2 f(x_C) + C \cdot h''(g(x_C))\g g(x_C)\g g(x_C)^T + C \cdot h'(g(x_C))\g^2g(x_C)$ then yields that
\[
y_C = -H_C^{-1}(h'(g(x_C))\g g(x_C)) ~.
\]
Also, note that \begin{equation}H_C \succ C \cdot h''(g(x_C))\g g(x_C)\g g(x_C)^T \label{eq:hlower} \end{equation} 
by convexity of $f, g$ and monotonicity and convexity of $h$. Finally, we can compute that
\[ \frac{dw(C)}{dC} = h'(g(x_C))+C \cdot h''(g(x_C))\g g(x_C)^T y_C = h'(g(x_C))\left(1-C \cdot h''(g(x_C))\|\g g(x_C)\|_{H_C^{-1}}^2\right) > 0 \] by monotonicity of $H$ and \cref{eq:hlower}, as desired.

\paragraph{Algorithm description and analysis.} Define \[ \delta_0 = \frac{\mu_f\eps_1^2}{100H_2Z_2^2L_g^4L_h^2} \text{ and } \eps_0 = \frac{\mu_f\eps_1^2}{100Z_2^2L_g^4L_h^2}. \] Note that $\delta_0, \eps_0 \le \eps_1/2.$ We compute a $C$ with $|C-C^*| \le \delta_0$ using binary search on the range $C \in [Z_1,Z_2]$ along with our ability to approximately compute $w(C)$. We may restrict our binary search to this range by \cref{eq:binsearchbound}. During each iteration of the binary search, we make an oracle call to compute $y_C$ with parameter $\eps = \eps_0$. We use this to compute an approximate value of $w(C) \approx C \cdot h'(g(y_C))$. We terminate our binary search when the upper and lower bounds are separated by $\delta_0.$ Therefore, our binary search requires $O\left(\log \frac{Z_2}{\delta_0} \right)$ iterations. The total runtime is $O\left( T(\eps_0) \log \frac{Z_2}{\delta_0} \right)$ as desired.

After termination of the binary search, the algorithm has outputted $y = y_C$ for some $C$. We now bound the difference between $C \cdot h'(g(y_C))$ and $C_0.$ Using that $f$ is $\mu_f$ strongly convex, the guarantee that \[ f(y_C) + C \cdot h(g(y_C)) \le f(x_C) + C \cdot h(g(x_C)) + \eps_0, \] and \cref{lemma:strong} gives us that \[ \|y_C - x_C\|_2 \le \left(\frac{2\eps_0}{\mu_f}\right)^\frac12. \] By the Lipschitz guarantees on $h'$ and $g$, we have that
\begin{equation} |w(C) - C \cdot h'(g(y_C))| = C|h'(g(y_C))-h'(g(x_C))| \le CL_gL_h\|y_C-x_C\|_2 \le Z_2L_gL_h\left(\frac{2\eps_0}{\mu_f}\right)^\frac12. \label{eq:error1} \end{equation} Therefore, we have that $C \cdot h'(g(y_C))$ is close to $w(C).$

We now bound $|w(C) - C_0| = |w(C) - w(C^*)|$, by using that $|C-C^*| \le \delta_0$ from the guarantees of the binary search. We start by bounding $\|x_C - x_{C^*}\|_2.$ Note that
\begin{align*}
&f(x_C) + C^* \cdot h(g(x_C)) = f(x_C) + C \cdot h(g(x_C)) + (C^*-C)h(g(x_C)) \\ &\le f(x_C) + C \cdot h(g(x_C)) + \delta_0 h(T_2) \le f(x_C) + C\cdot h(g(x_C)) + \delta_0H_2 \\ &\le f(x_{C^*}) + C\cdot h(g(x_{C^*})) + \delta_0H_2 \\ &= f(x_{C^*}) + C^*\cdot h(g(x_{C^*})) + (C-C^*) \cdot h(g(x_{C^*})) + \delta_0 H_2 \\ &\le f(x_{C^*}) + C^*\cdot h(g(x_{C^*})) + 2\delta_0 H_2.
\end{align*}
Therefore, $x_C$ is an additive $2\delta_0 H_2$ minimizer to the objective $f(x) + C^* \cdot h(g(x)).$ Because $f$ is $\mu_f$ strongly convex and \cref{lemma:strong} we have that
\[ \|x_C - x_{C^*}\|_2 \le \left(\frac{4\delta_0 H_2}{\mu_f}\right)^\frac12. \] Now we can compute that
\begin{align}
&|w(C) - C_0| = |w(C) - w(C^*)| = |C \cdot h'(g(x_C)) - C^* \cdot h'(g(x_{C^*}))| \nonumber \\
&\le |(C-C^*)h'(g(x_C))| + |C^*(h'(g(x_C))-h'(g(x_{C^*})))| \le \delta_0 H_2 + Z_2 L_hL_g \|x_C - x_{C^*}\|_2 \nonumber \\
&\le \delta_0 H_2 + Z_2 L_hL_g \left(\frac{4\delta_0 H_2}{\mu_f}\right)^\frac12. \label{eq:error2}
\end{align}
Combining \cref{eq:error1,eq:error2} gives us that \[ |C \cdot h'(g(y_C)) - C_0| \le Z_2L_gL_h\left(\frac{2\eps_0}{\mu_f}\right)^\frac12 + \delta_0 H_2 + Z_2 L_hL_g \left(\frac{4\delta_0 H_2}{\mu_f}\right)^\frac12 \le \frac{\eps_1}{2L_g} \] by our choices of $\delta_0$ and $\eps_0$.
\begin{align*}
&\|\g f(y_C) + C_0 \g g(y_C)\|_2 \le \|\g f(y_C) + C \cdot h'(g(y_C)) \g g(y_C)\|_2  + |C \cdot h'(g(y_C)) - C_0| \|\g g(y_C)\|_2 \\
&\le \eps_0 + \frac{\eps_1}{2L_g} \|\g g(y_C)\|_2 \le \eps_0 + \frac{\eps_1}{2} \le \eps_1
\end{align*}
as $\|g(y_C)\|_2 \le L_g$ by \cref{lemma:gradlip}, and $\eps_0 \le \eps_1/2.$ This completes the proof.
\end{proof}
\begin{proof}[Proof of \cref{thm:l2lp}]
We split our proof into parts.

\paragraph{Reduction for resistances and demand vector.} Recall that our objective is \[ \min_{B^Tf=d} f^TRf + \|f\|_p^2, \] where all parameters in $d$ and $r$ are bounded by $2^{\poly\log(m)}.$ We first reduce to the case where $R \se \err I.$ To do this, pick a small parameter $\nu = \err$, increase all diagonal entries in the matrix $R$ by $\nu.$ This can only affect our optimum by an additive $\err$ and this is acceptable in \cref{thm:l2lp}. Thus, from now on we assume $R \se \nu I.$ We also reduce to the case where the demand $d$ has some element with absolute value at least $\err.$ To do so, change $d \to d + \nu_2 \chi_{ab}$, for some small parameter $\nu_2 = \err$. This once affects the optimum by $\err$. In this way, we assume that without loss of generality that $d_a \ge \nu_2.$

\paragraph{Reducing to unconstrained problem.} In order to apply \cref{lemma:reduc}, we need to make our objective unconstrained. Let $f_0$ be any flow with $B^Tf_0 = d$ and let $P$ be a $(m-n+1) \times m$ matrix such that the map $x \to Px+f_0$ is an isomorphism from $\R^{m-n+1}$ to the set of flows $f \in \R^E$ with $B^Tf=d.$ We choose $P$ specifically as follows. Let $T$ be a spanning tree in $G$, and we view vector $x \in \R^{m-n+1}$ as the set of flows on edges in $E\bs E(T)$, i.e. the off tree edges. Now, $P$ is defined so that $Px$ has the same flows on off tree edges as $x$, and the flows for edges in $T$ are chosen in the unique way so that $Px$ is a circulation. Note that all entries are $P$ are polynomially bounded.

\paragraph{Choice of parameters for \cref{lemma:reduc}.} Then we can write our objective as \[ \min_{x \in \R^{m-n+1}} (Px+f_0)^TR(Px+f_0) + \|Px+f_0\|_p^2, \] and we can set the functions $f, g, h$ in \cref{lemma:reduc} as \[ f(x) = (Px+f_0)^TR(Px+f_0), g(x) = \|Px+f_0\|_p^2, h(x) = x^{p/2}. \] We now describe our choices for the constants in the statement of \cref{lemma:reduc}. We choose $C_0 = 1$. We choose $\X$ to be the open region such that all entries of $x$ are bounded by $2^{\poly\log(m)}.$ We can assume this because all optimal flows should have no cycles, and all entries of $d$ are bounded by $2^{\poly\log(m)}.$ Note that \[ \g^2 f(x) = P^TRP \se \nu P^TP \se \nu I \] by our construction of $P$. Therefore, $f(x)$ is $\nu$ strongly convex by \cref{lemma:equiv}, and we set $\mu_f = \nu.$ As $d_a \ge \nu_2$, every flow $Px+f_0$ has some entry which is at least $\nu_2/m.$ This implies that $f(x) \ge \frac{\nu\nu_2^2}{m^2}$ for all $x$ and $g(x) \ge \frac{\nu_2^2}{m^2}$ for all $x$. Thus we set $T_1 = \frac{\nu\nu_2^2}{m^2}.$ By our choice of $\X, P$ and that all entries of $R$ are $2^{\poly\log(m)}$, we have that for all $x \in \X$ that $f(x) \le 2^{\poly\log(m)}$ and $g(x) \le 2^{\poly\log(m)}$. Thus we can set $T_2 = 2^{\poly\log(m)}$.

By our choice of $\X, P$ we can set $L_g = 2^{\poly\log(m)}.$ We can set $L_h = \frac{p}{2} \cdot h'(T_2) = \frac{p}{2} \cdot T_2^\frac{p-2}{2}.$ We set $Z_1 = 0$ and $Z_2 = \frac{C_0}{h'(T_1)} = \frac{C_0}{\frac{p}{2} \cdot T_1^\frac{p-2}{2}}$. We set $H_1 = h(T_1) = T_1^{p/2}$, and $H_2 = h(T_2) = T_2^{p/2}.$ Finally, we set $\eps_1 = \err.$ These choices of parameters satisfy all properties. Additionally, we have that
\[ \log\max\{L_g, L_h, Z_2, H_2\} = \O(1) \text{ and } \log\min\{\mu_f, \eps_1\} = -\O(1) \] so that in \cref{eq:timebound} we have
\[ \log \frac{\mu_f\eps_1^2}{100Z_2^2L_g^4L_h^2} = -\O(1) \text{ and } \log \frac{H_2Z_2^3L_g^4L_h^2}{\mu_f\eps_1^2} = \O(1). \]

\paragraph{Runtime bound of $m^{1+o(1)}$.} This follows from the statement of \cref{thm:l2lpp}, the fact that $\log \eps_1 \ge -\O(1)$ and $\log \frac{Z_2}{\delta_0} \le \O(1)$, and the fact that all entries of $R$ and $d$ are bounded by $2^{\poly\log(m)}$.

\paragraph{\cref{thm:l2lpp} satisfies \cref{eq:thing1,eq:thing2}.} Let $L_2$ be the smoothness of the function \[ A(x) \defeq f(x) + C \cdot h(g(x)) = (Px+f_0)^TR(Px+f_0) + C \cdot \|Px+f_0\|_p^p. \] We show that $L_2 = 2^{\poly\log(m)}.$ Indeed, we can compute that \[ \g^2 A(x) = P^TRP + p(p-1)C \cdot P^T|Px+f_0|^{p-2}P \pe 2^{\poly\log(m)}I \] as
all entries of $P$ are polynomially bounded, all entries of $R$ are $2^{\poly\log(m)}$, and our restriction to $\X$ implies that all entries of $Px+f_0$ are $2^{\poly\log(m)}.$ Hence $L_2 = 2^{\poly\log(m)}$ by \cref{lemma:equiv}.

\cref{eq:thing2} follows directly from the statement of \cref{thm:l2lpp}. Because the objective is $L_2$ smooth, \cref{lemma:smoothgrad} gives us
\begin{align*}
&\|\g f(y_C) + C \cdot h'(g(y_C)) \cdot \g g(y_C)\|_2 = \| \g A(y_C) \|_2 \\
&\le (2L_2(A(y_C)-A(x_C)))^\frac12.
\end{align*}
Thus, if $A(y_C) - A(x_C) \le \frac{\eps^2}{2L_2}$, then \[ \|\g f(y_C) + C \cdot h'(g(y_C)) \cdot \g g(y_C)\|_2 \le \eps \] as desired, as $\log 1/\eps = \O(1)$ and $\log L_2 = \O(1).$

\paragraph{Finishing the proof.} Define $B(x) = f(x) + C \cdot g(x).$ By \cref{lemma:reduc}, we can compute a $y$ with $\| \g B(y) \|_2 \le \eps_1.$ Therefore, we have that \[ B(y)-B(x^*) \le \g B(y)^T(y-x^*) \le \|\g B(y)\|_2 \|y-x^*\|_2 \le \eps_1\|y-x^*\|_2. \] Because all entries of demand vector $d$ are bounded by $2^{\poly\log(m)}$, we have that both $y$ and $x^*$ have all entries bounded by $2^{\poly\log(m)}.$ Therefore, for $\eps_1 = \err$ we have that $B(y)-B(x^*) \le \eps_1\|y-x^*\|_2 \le \err$. From our definitions of $f, g$, this is exactly the desired result of \cref{thm:l2lp}.
\end{proof}

\section{Proofs of Reweighting Techniques from \cite{Madry16}}
\label{sec:madryweightproofs}
\subsection{Proof of \cref{lemma:weightvsres}}
\label{proofs:weightvsres}
\begin{proof}
It suffices to describe the weight changes edge by edge. We also perform weight changes to make the resistance of edge $e$ \emph{at least} $r_e+r_e'$, and this implies the desired result. For edge $e$, as we have assumed that our original flow problem is undirected, we have that $1 \le \up_e = \um_e \le U.$ Without loss of generality, assume that $\up_e-f_e \le \um_e+f_e.$ Set $(\wpe)' = r_e'(\up_e-f_e)^2$ and $(\wme)' = \frac{(\wpe)'(\um_e+f_e)}{\up_e-f_e}$. We show that this satisfies the desired conditions.

Because $\frac{(\wpe)'}{\up_e-f_e} = \frac{(\wme)'}{\um_e+f_e}$ for all $e$, we know that if edge $e$ was $\zeta_e$ coupled for weights $w$ then it also $\zeta_e$ coupled for weights $w+w'.$

Note that \[ (\wme)' = \frac{(\wpe)'(\um_e+f_e)}{\up_e-f_e} \le r_e'(\um_e+f_e)(\up_e-f_e) \le 2Ur_e's_e \] by the definition of $s_e.$ Therefore, $(\wpe)' + (\wme)' \le 2(\wme)' \le 4Ur_e's_e.$

Finally, note that the new resistances are at least \[ \frac{\wpe+(\wpe)'}{(\up_e-f_e)^2} + \frac{\wme+(\wme)'}{(\ume+f_e)^2} \ge r_e + \frac{(\wpe)'}{(\up_e-f_e)^2} \ge r_e+r_e' \] by the definition of $(\wpe)'.$
\end{proof}

\subsection{Proof of \cref{lemma:resbound}}
\label{proofs:resbound}
\begin{proof}
Consider an edge $e = (u,v).$ Without loss of generality, we can assume that $\upe-f_e\le \ume+f_e$, so that $s_e = \upe-f_e$ and $\rpe \ge \rme.$ Let $\phi$ be voltages induced by the electric flow $\hf.$ We can write by Ohm's law that
\[ s_e^{-1} = \frac{1}{\upe-f_e} = \frac{\frac{|\hf_e|}{(\upe-f_e)^2}}{\frac{|\hf_e|}{\upe-f_e}} \le \frac{r_e\hf_e}{\rpe} = \frac{|\phi_u-\phi_v|}{\rpe} \le \frac{|\phi_a-\phi_b|}{\rpe} = \frac{\normrho_{w,2}^2}{\rpe} \]
 where we have used \cref{lemma:rhoenergy}. Here, the second to last inequality follows from the well-known fact that for an electric $ab$-flow that $\phi_a \le \phi_v \le \phi_b$ for all vertices $v$. The last equality follows from the fact that for the electric $\chi$-flow $\hf$ induced by $\phi$ that \[ \phi_b - \phi_a = \chi^T\phi = \hf^TB\phi = \hf^TR\hf = \E_r(\hf) = \|\rho\|_{w,2}^2 \] by \cref{lemma:rhoenergy}.
\end{proof}

\subsection{Proof of \cref{lemma:energyboost}}
\label{proofs:energyboost}
\begin{proof}
Let $\phi$ be the potentials induced by the electric flow with resistances $r$, and let $\phi' = \phi/\E_r(\hf_r).$ We use $\phi'$ as a certificate in \cref{lemma:energydual} to lower bound $\E_{r+r'}(\hf_{r+r'}).$ Specifically, by \cref{lemma:energydual} we have that
\begin{align*}
\frac{1}{\E_{r+r'}(\hf_{r+r'})} &\le \sum_{e \in E} \frac{(\phi'_u-\phi'_v)^2}{r_e+r_e'} = \sum_{e\in E} \frac{(\phi'_u-\phi'_v)^2}{r_e}-\frac{(\phi'_u-\phi'_v)^2 r_e'}{r_e(r_e+r'_e)} \\
&\le \frac{1}{\E_r(\hf_r)} - \sum_{e \in E} \frac{(\phi'_u-\phi'_v)^2 r_e'}{2r_e^2} = \frac{1}{\E_r(\hf_r)} - \frac{1}{2\E_r(\hf_r)^2} \sum_{e\in E} r_e'[\hf_r]_e^2,
\end{align*}
where we have used Ohm's Law and the definition of $\phi'$. Rearranging, taking logarithms, and using $\log(1-x) \le -x$ gives us
\begin{align*}
\log \E_{r+r'}(\hf_{r+r'}) - \log \E_r(\hf_r) &= -\log \frac{\E_r(\hf_r)}{\E_{r+r'}(\hf_{r+r'})} \ge -\log\left(1 - \sum_{e\in E} \frac{r_e'[\hf_r]_e^2}{2\E_r(\hf_r)} \right) \\
&\ge \sum_{e \in E} \frac{r_e'[\hf_r]_e^2}{2\E_r(\hf_r)} = \frac12 \sum_{e \in E} \frac{r_e's_e^2\max\{\rpe,\rme\}^2}{\normrho_{w,2}^2}
\end{align*}
by the definitions of $s_e, r_e, \rho$, and \cref{lemma:rhoenergy}.
\end{proof}

%% file: approximate.tex
\section{Discussion of Approximate Solvers}
\label{sec:approximate}

In \cref{sec:weight}, some of our lemmas we proven assuming the solvers of \cref{thm:l2lpp,thm:l2lp} were exact. We discuss how to adapt the proofs to handle the approximate nature of the solvers here.

We first must show that the resistances are polynomially bounded on the central path.
\begin{lemma}
\label{lemma:polyres}
Let $(f,y,w)$ be a $\frac{1}{100}$-coupled point for parameter $t$, and $F_t \ge m^{1/2-\eta}.$ Then for all edges $e = (u,v)$ we have that $|y_v-y_u| \le 3m^2$ and $r_e \le 10m^4.$
\end{lemma}
\begin{proof}
Consider a decomposition of flow $f$ into cycles and $ab$-paths that do not cancel. For edge $e$, WLOG $f_e \ge 0$. We know that edge $e$ is either involved in a cycle of a path in the decomposition. We distinguish these two cases.
\paragraph{Edge $e$ in a cycle.} Let the edges in the cycle be $e = e_1, e_2, \cdots, e_k$ in order, so that $f_{e_i} \ge 0$ for all $1 \le i \le k$ and there exist vertices $v_1, \cdots, v_k, v_{k+1} = v_1$ such that edge $e = (v_i, v_{i+1}).$ For all $1 \le i \le k$ because $(f,y)$ is $\frac{1}{100}$-coupled we have that
\begin{align*}
y_{v_{i+1}} - y_{v_i} &\ge \frac{\wpe}{\upe-f_{e_i}}-\frac{\wme}{\ume+f_{e_i}} - \frac{1}{100}r_e^{-\frac12} \\
					  &\ge \frac{\wpe}{\upe-f_{e_i}}-\frac{\wme}{\ume+f_{e_i}} - \frac{1}{100}\left(\frac{\wpe}{\upe-f_{e_i}}+\frac{\wme}{\ume+f_{e_i}}\right) \\
					  &\ge -\frac{\frac{101}{100}\wme}{\ume+f_{e_i}} \ge -2m,
\end{align*}
where we have used that $\wme \le m+1$ as $\|w\|_1 \le 3m$ and $U \ge 1.$ Therefore, we have that for edge $e = e_1$ that $y_{v_2}-y_{v_1} \ge -2m.$
Also, we have that
\[ y_{v_2}-y_{v_1} = -\sum_{i=2}^k (y_{v_{i+1}} - y_{v_i}) \le 2m^2 \] by the previous bound. Hence $|y_{v_2}-y_{v_1}| \le 2m^2$.
\paragraph{Edge $e$ in an $ab$-path.} Let the edges in the path be $e_1, e_2, \cdots, e_k$, so that $f_{e_i} \ge 0$ for all $1 \le i \le k.$ Additionally, there exist vertices $v_0, v_1, \cdots, v_k$ such that $a = v_0, b = v_k$ and $e_i = (v_{i-1}, v_i).$ Let $e = e_{\ell}$ for some index $\ell.$ By the same argument in the above paragraph, we have that $y_{v_i}-y_{v_{i-1}} \ge -2m.$ Additionally, we have that
\begin{align*}
y_{v_\ell}-y_{v_{\ell-1}} &= -\left(\sum_{i=1}^{\ell-1} (y_{v_i}-y_{v_{i-1}}) + \sum_{i=\ell+1}^k (y_{v_i}-y_{v_{i-1}}) + (y_a-y_b)\right) \\
						  &\le 2m^2 + \chi^Ty \le 2m^2 + \frac{4\|w\|_1}{F_t} \le 3m^2,
\end{align*}
where we have used \cref{lemma:chiy} and $F_t \ge m^{1/2-\eta}.$

Now we bound $r_e$ for edge $e = (u,v)$. WLOG $f_e \ge 0.$ We have that
\begin{align*}
3m^2 &\ge y_v-y_u \ge \frac{\wpe}{\upe-f_e}-\frac{\wme}{\ume+f_e} - \frac{1}{100}r_e^{-\frac12} \\
	 &\ge \frac{\wpe}{\upe-f_e}-\frac{\wme}{\ume+f_e} - \frac{1}{100}\left(\frac{\wpe}{\upe-f_e}+\frac{\wme}{\ume+f_e}\right) \\
	 &= \frac{99}{100}\left(\frac{\wpe}{\upe-f_e}+\frac{\wme}{\ume+f_e}\right) - \frac{2\wme}{\ume+f_e} \ge \frac{99}{100}r_e^\frac12 - 4m.
\end{align*}
Now rearranging gives us $r_e \le 10m^4$ as desired.
\end{proof}

We discuss \cref{lemma:energymax,lemma:gconcave,lemma:rhoinf} in the order that they were proven in \cref{sec:weight}.
\subsection{Discussion of \cref{lemma:energymax}}
If $OPT \le \err$, the result is trivial, as the statement of \cref{lemma:energymax} allows us $\err$ additive error. So we assume $OPT \ge \err.$ As in the proof of \cref{lemma:energymax}, define $f_1$ to be an $\err$-approximate minimizer to \cref{eq:expression}, and define $r' = W\frac{f_1^{2(p-1)}}{\|f_1^2\|_p^{p-1}}.$ Define $r_2 = r+r'$. Note that \[ f_1^TR_2f_1 = f_1^TRf_1 + W\|f_1\|_p^2 \ge OPT. \]

We now show that there exists a $\phi$ such that $\|B\phi - R_2f_1\|_2 \le \err.$ Let $f^*$ be the exact optimal flow and let $r_2^* = r^*+r$ denote the induced resistances from $f^*.$ As in the proof of \cref{lemma:energymax} in \cref{sec:weight}, there is a $\phi^*$ with $B\phi^* = R_2^*f^*.$ Note that all $r_e \ge \frac{1}{U^2}$, hence the objective $f^TRf + W\|f\|_p^2$ is $\frac{1}{U^2}$ strongly convex by \cref{lemma:equiv}. Therefore, if $f_1$ is a $\err$-approximate minimizer to \cref{eq:expression}, then \cref{lemma:strong} gives us that \[ \|f^*-f_1\|_2 \le \poly(m) \cdot \err = \err. \] The definition of $r_2$ shows that $\|r_2-r_2^*\|_2 \le \err$ as well. As all values of $f_1$ and $r_2$ are polynomially bounded, we have that $\|R_2f_1 - R_2^*f^*\|_2 \le \err$ as well. Finally, this shows that
\[ \|B\phi^* - R_2f_1\|_2 = \|R_2f_1 - R_2^*f^*\|_2 \le \err \] as desired.

Define $v = B\phi-R_2f_1$, and let $\|v\|_2 = \d$ where $\d \le \err.$ \cref{lemma:energydual} we have that
\begin{align*}
\E_{r_2}(\hf_{r_2}) &\ge \frac{(\phi^T\chi)^2}{(B\phi)^TR_2^{-1}(B\phi)} = \frac{\left(\phi^T(B^Tf_1)\right)^2}{(B\phi)^TR_2^{-1}(B\phi)} = \frac{\left((B\phi)^Tf_1\right)^2}{(B\phi)^TR_2^{-1}(B\phi)} \\ &= \frac{\left(f_1^TR_2f_1 + v^Tf_1\right)^2}{f_1^TR_2f_1 + 2v^Tf_1 + v^TR_2^{-1}v} \ge \frac{(OPT - \d\|f_1\|_2)^2}{OPT + 2\d\|f_1\|_2 + \d^2\|R_2^{-1}\|_\infty} \\ &\ge OPT - \err
\end{align*}
where we have used that $\d \le \err, OPT \ge \err, \|f_1\|_2 \le \poly(m)$, and $\|R_2^{-1}\|_\infty \le \poly(m)$ as $r_e \ge \frac{1}{U^2}$ for all edges $e$.

\subsection{Discussion of \cref{lemma:rhoinf}}
One change is needed. The resistances $r'$ and weights $w'$ computed are only additive $\err$ maximizers for the objective for $g_q(W)$. However, this is sufficient to adapt \cref{eq:adapt} to conclude that \[ g_q(W+W')-g_q(W) \ge \frac12\left(\frac{\max(|\rpe|,|\rme|)^2 r''_e/s_e}{\normrho_{w^\new,2}^2}\right) - \err. \] We can propagate this change throughout the proof of \cref{lemma:rhoinf}.

\subsection{Discussion of $\frac{1}{\poly(m)}$ versus $\err$ error in \cref{thm:l2lpp}}
\label{sec:KPSWpoly}
\cref{thm:l2lpp} was stated with errors $\frac{1}{\poly(m)}$ in Theorem 1.1 in \cite{KPSW19}. We explain why those errors can in fact be made to be $\err$.
Consider Theorem 3.7 in \cite{KPSW19}, and the proof of Theorem 1.1 using it. In \cite{KPSW19}, they set the parameter $\delta = \frac{1}{\poly(m)}$, and use algorithm \RecursivePreconditioning~$O(\log m)$ times. Instead, we set $\delta = \err$ and apply algorithm \RecursivePreconditioning~$\O(1)$ times. Our choice of $\d$ satisfies the conditions of Theorem 3.7 in \cite{KPSW19} as $\log\frac{1}{\delta} = \O(1)$.

%% file: numerical.tex
\section{Discussion of Numerical Issues}
\label{sec:numerical}
We briefly discuss issues of numerical stability in our algorithm.
\subsection{Laplacian Solvers}
As all resistances, duals, and flows are polynomially bounded by \cref{lemma:polyres}, the accuracy of Laplacian solvers in \cref{thm:lap} is sufficient for our algorithms. See discussion in \cite{Madry13} for further details.

\subsection{Interior Point Methods}
For readability, we have written \cref{lemma:progress} to start from a $0$-coupled point $(f,y,w)$. Here we justify this assumption and show how to correct the analysis if instead the point $(f,y,w)$ was $\frac{1}{\poly(m)}$-coupled, and how this change propagates throughout the proof.

Let $(f,y,w)$ be a point such that edge $e$ is $\zeta_e$-coupled for all $e$, and $|\zeta_e| \le \frac{1}{\poly(m)}$ for all $e$. As all $r_e \ge \frac{1}{U}$, we have that $\|\zeta\|_{R^{-1}} \le \frac{1}{\poly(m)}$ as well, so that $(f,y)$ is $\frac{1}{\poly(m)}$-coupled. It is direct to adapt the statement of \cref{lemma:progress} to instead get that edge $e$ is $5\zeta_e + 5\d^2\left(\frac{\wpe|\rpe|^2}{\upe-f_e} + \frac{\wme|\rme|^2}{\ume+f_e}\right)$-coupled for all edges $e$.

This change propagates to the results of \cref{sec:finalweight}. \cref{lemma:perfectweight} remains unchanged. In the weight reduction procedure in \cref{lemma:weightreduce}, all guarantees are the same, except that edge $e$ will now be $5\zeta_e + 5\d^2\left(\frac{(\wpe)^\new|\rpe|^2}{\upe-f_e} + \frac{(\wme)^\new|\rme|^2}{\ume+f_e}\right)$-coupled. When we propagate this change to \cref{lemma:finalweight} this results in the point $(f^\new, y^\new, w^{\new_3})$ to instead be
$O(m^4U\|\zeta\|_{R^{-1}}) + \frac{1}{100}$-coupled. This is because the resistances may change during the algorithm, but the final resistances $R^\new$ satisfy $\frac{1}{O(m^4U)}R \pe R^\new \pe O(m^4U)R$ by \cref{lemma:polyres}, and the fact that $r_e \ge \frac{1}{U}$ for all $e$. Therefore, the method can tolerate errors of $\frac{1}{\poly(m)}$.